\documentclass[journal]{IEEEtran}

\usepackage{amsfonts,amsmath,amssymb,amsthm}
\usepackage{bbm,bm}
\usepackage{graphicx,epsfig,color}
\usepackage{algorithm,algorithmic}
\usepackage{comment}
\usepackage{epstopdf}
\usepackage{subfigure}

\newtheorem{thm}{Theorem}

\newtheorem{lem}{Lemma}
\newtheorem{Proposition}{Proposition}

\newtheorem{defn}{Definition}

\newcommand{\red}{\textcolor{red}}
\newcommand{\eg}{{\it e.g., }}

\newcommand{\ie}{{\it i.e., }}
\newcommand{\Ie}{{\it I.e., }}
\newcommand{\E}{{\mathbb E }}

\newcommand{\GG}{{\mathcal G}}
\newcommand{\VV}{{\mathcal V}}
\newcommand{\NN}{{\mathcal N}}
\newcommand{\EE}{{\mathcal E}}
\newcommand{\DD}{{\mathcal D}}
\newcommand{\CC}{{\mathcal C}}
\newcommand{\TT}{{\mathcal T}}
\newcommand{\SSS}{{\mathcal S}}

\begin{document}
\title{Optimal Control of Distributed Computing Networks with Mixed-Cast Traffic Flows}
\author{Jianan Zhang, Abhishek Sinha, Jaime Llorca, Antonia Tulino, Eytan Modiano
\thanks{Part of the material in this paper was presented at IEEE International Conference on Computer Communications (INFOCOM), 2018.}
\thanks{J. Zhang (jianan@mit.edu) and E. Modiano (modiano@mit.edu) are with the Laboratory for Information and Decision Systems, Massachusetts Institute of Technology, Cambridge, MA 02139, USA. A. Sinha (abhishek.sinha@iitm.ac.in) is with the Department of Electrical Engineering, IIT Madras. J. Llorca (jaime.llorca@nokia-bell-labs.com) and A. Tulino (a.tulino@nokia-bell-labs.com) are with Nokia Bell Labs, Holmdel, NJ 07733, USA. A. Tulino is also with the University degli Studi di Napoli Federico II, 80138 Naples, Italy.}
\thanks{This work was supported by DTRA grants HDTRA1-13-1-0021 and HDTRA1-14-1-0058, and by NSF grant number CNS-1617091.}}

\maketitle

\begin{abstract}
Distributed computing networks, tasked with both packet transmission and processing, require the joint optimization of communication and computation resources. We develop a dynamic control policy that determines both routes and processing locations for packets upon their arrival at a distributed computing network. The proposed policy, referred to as Universal Computing Network Control (UCNC), guarantees that packets i) are processed by a specified chain of service functions, ii) follow cycle-free routes between consecutive functions, and iii) are delivered to their corresponding set of destinations via proper packet duplications.
UCNC is shown to be throughput-optimal for any mix of unicast and multicast traffic, and is the first throughput-optimal policy for non-unicast traffic in distributed computing networks with both communication and computation constraints. Moreover, simulation results suggest that UCNC yields substantially lower average packet delay compared with existing control policies for unicast traffic. 
\end{abstract}

\section{Introduction}

The recent convergence of IP networks and IT clouds is fueling the emergence of large-scale {\em distributed computing networks}
that can host content and applications close to information sources and end users, providing rapid response, analysis, and delivery of augmented information in real time \cite{fxbook}.
This, in turn, enables a new breed of services, often referred to as {\em augmented information services}.
Unlike traditional information services, in which users consume information that is produced or stored at a given source and is delivered via a
communications network, augmented information services provide end users with information that results from the  {\em real-time processing} of source data flows via possibly multiple service functions that can be hosted at multiple locations in a distributed computing network. 

Particularly popular among these services is the class of {\bf automation services}, in which information sourced at sensing devices in physical infrastructures such as homes, offices, factories, and cities, is processed in real time in order to deliver instructions that optimize and control the automated operation of physical systems. Examples include industrial internet services (e.g., smart factories), automated transportation, smart buildings, smart homes, etc \cite{industrial_internet}. Also gaining increasing attention is the class of {\bf augmented experience services}, which allow users to consume multimedia streams that result from the combination of multiple live sources and contextual information of real-time relevance. Examples include telepresence, real-time computer vision, virtual classrooms/labs/offices, and augmented/virtual reality \cite{ar_2013}.
In addition to application-level services, with the advent of network functions virtualization (NFV), {\bf network services} that typically run on dedicated hardware can also be implemented in the form of software functions running on general purpose servers distributed throughout a computing network. Software defined networking (SDN) technologies can then be used to steer network flows through the appropriate chain of network functions \cite{fxbook}.

While most of today's computationally intensive services are hosted at centralized cloud data centers, the increasingly low latency requirements of next generation services are driving cloud resources closer to the end users in the form of small cloud nodes at the edge of the network, resulting in what is referred to as a distributed cloud network or distributed computing network \cite{fxbook}.
Compared to traditional centralized clouds, distributed computing networks provide increased flexibility in the allocation of computation and network resources, and a clear advantage in meeting stringent service latency, mobility, and location-awareness constraints.

To maximize the benefits of this attractive scenario and enable its sustainable growth, operators must be able to dynamically control the configuration of a diverse set of services according to changing demands, while minimizing the use of the shared physical infrastructure.
A key aspect driving both performance and efficiency is the actual placement of the service functions, as well as the routing of network flows through the appropriate function instances.
Traditional information services have addressed the efficient flow of information from data sources to destinations,
where sources may include static processing elements, mostly based on rigid hardware deployments.
In contrast, the efficient delivery of next generation services requires {\em jointly optimizing where to execute each service function and how to route network flows in order to satisfy service demands that may be of unicast or multicast nature.}

The {\bf static} service function placement and routing problems have been studied in previous literature.
Given fixed service rates, 
linear programming formulations for joint function placement and unicast routing under maximum flow or minimum cost objectives were developed in \cite{charikar2014multi, csdp_icc15, bari2015orchestrating, barcelo2016}.
Under fixed routing, 
algorithms for function placement with bi-criteria approximation guarantees were developed in \cite{Cohen15}.
Under fixed function placement, approximation algorithms for unicast traffic steering were given in \cite{Cao2017}.
Approximation algorithms for joint function placement and unicast routing were developed for a single function per flow in \cite{charikar2014multi} and for service function chains in \cite{nsdp_info17,Kuo2017}.

The study of {\bf dynamic} control policies for service function chains was initiated in \cite{dcnc_info16, dcnc_icc16}. The authors developed throughput-optimal policies to jointly determine processing locations and routes for unicast traffic flows in a distributed computing network, based on the backpressure algorithm. Another backpressure-based algorithm was developed in \cite{Destounis2016} in order to maximize the rate of queries for a computation operation on remote data from a particular destination.

However, no previous work has addressed the network computation problem under non-unicast traffic. In fact, it was only very recently that the first throughput-optimal algorithm for generalized flow (any mix of unicast and multicast traffic) problems in communication networks was developed \cite{umw_info17}.
Given that internet traffic is increasingly a diverse mix of unicast and multicast flows, in this work, we address the design of {\em throughput-optimal dynamic packet 
processing and routing policies for mixed-cast (unicast and multicast) service chains in distributed computing networks.} 
Our solution extends the recently developed universal max-weight algorithm \cite{umw_info17} to handle both communication and computation constraints in a distributed computing network.
Our proposed control policy also handles flow scaling, a prominent characteristic of traffic flows in distributed computing networks, where a flow may expand or shrink due to service function processing.\footnote{Video transcoding is an example of a service function that changes flow size.} A preliminary version of this paper was presented in \cite{conf}.

Our contributions are summarized as follows:
\begin{itemize}
\item
We characterize the capacity region of a distributed computing network hosting an arbitrary set of service function chains that can process an arbitrary mix of unicast and multicast traffic. Such first characterization involves the definition of generalized flow conservation laws that capture flow chaining and scaling, due to service function processing, and packet duplication, due to multicasting. 
\item
We develop a universal control policy for service function chains in distributed computing networks, referred to as Universal Computing Network Control (UCNC). 
UCNC determines both routes and processing locations for packets upon their arrival at a distributed computing network, and guarantees that packets i) are processed by a specified chain of service functions, ii) follow cycle-free routes between consecutive functions, and iii) are delivered to their corresponding set of destinations via proper packet duplications.
\item UCNC is shown to be throughput-optimal for any mix of unicast and multicast traffic, and is the first throughput-optimal algorithm for non-unicast traffic in distributed computing networks. Even for unicast traffic, compared with the previous throughput-optimal algorithm \cite{dcnc_icc16}, UCNC yields much shorter average packet delay. 
\end{itemize}

The rest of the paper is organized as follows. We introduce the model in Section \ref{sc:model}, and characterize the capacity region in Section \ref{sc:capacity}. In Section \ref{sc:routing}, we develop a routing policy to stabilize a virtual queuing system. 
In Section \ref{sc:physical}, we prove that the same routing policy, along with a proper packet scheduling policy, is throughput-optimal for the associated computing network. Section \ref{sc:simulation} presents numerical simulations, Section \ref{sc:extensions} presents extensions, and Section \ref{sc:conclusion} presents concluding remarks.

\section{System Model}
\label{sc:model}
In this section, we present models for distributed computing networks, service function chains, and mixed-cast traffic.
\subsection{Computing network model}

We consider a distributed computing network modeled as a directed graph $\mathcal G=(\mathcal V,\mathcal E)$ with $n=|\mathcal V|$ nodes and $m=|\mathcal E|$ links.
A node may represent a router, which can forward packets to neighboring nodes, or a distributed computing location, which, in addition, can host service functions for flow processing.
When network flows go through a service function at a computation node, they consume computation resources (\eg CPUs).
We denote 
by $\mu_u$ the processing capacity of node $u\in\mathcal V$.
A link represents a network connection between two nodes.
When network flows go through a link, they consume communication resources (\eg bandwidth).
We denote 
by $\mu_{uv}$ the transmission capacity of link $(u,v)\in\mathcal E$.


\subsection{Service model}
A service $\phi \in \Phi$ is described by a chain of $M_\phi$ functions $(\phi,i)$, $i \in \{1,\dots,M_\phi\}$. Each function $(\phi, i)$ is characterized by its computation requirement  
$r^{(\phi,i)}$, indicating that $r^{(\phi,i)}$ computation resource units are required to process a unit input flow. 
Function $(\phi, i)$ is also characterized by a flow scaling factor $\xi^{(\phi,i)}$, indicating that the average flow rate at the output of function $(\phi,i)$ is $\xi^{(\phi,i)}$ times the average input flow rate.
We assume that function $(\phi,i)$ is available at a subset of computation nodes 
$\NN_{(\phi,i)}\subseteq \VV$.
A flow that requires service $\phi$ must be processed by the functions $(\phi,i)$, $i \in \{1,\dots,M_\phi\}$ in order.

Figure \ref{fig:sfc} illustrates an example of a service function chain for video streaming. The first function in the chain is a firewall, with 
computation requirement $r^{(\phi,1)} = 0.1$ and 
flow scaling $\xi^{(\phi,1)} = 1$. 
The second function in the chain is a transcoding function, with 
computation requirement $r^{(\phi,2)} = 2$ and 
flow scaling $\xi^{(\phi,2)} = 0.8$. 
The numbers above the links indicate the flow rates at each stage of the service chain, and the numbers above the functions indicate the computation rates required to process the incoming flow.


\begin{figure}[h]
\centering
\includegraphics[width=.75\linewidth]{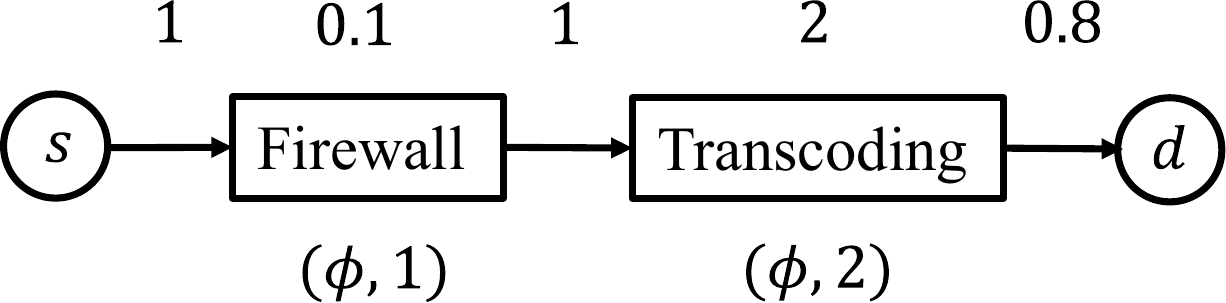}
\caption{An illustration of a service function chain with different function computation requirements and flow scaling.}
\label{fig:sfc}
\end{figure}
\subsection{Traffic model}


A commodity-$(c,\phi)$ flow is specified by a source node $s_{c}$, a set of destination nodes $\DD_{c}$, and a service $\phi$. Packets of commodity-$(c,\phi)$ flow enter the network at $s_{c}$ and exit the network for consumption at $\DD_{c}$ after being processed by the service functions in $\phi$. 
A flow is {\em unicast} if $\mathcal D_{c}$ contains a single node in $\mathcal V$, denoted by $d_c$, and is {\em multicast} if $\mathcal D_c$ contains more than one node in $\mathcal V$.
We denote by $(\CC, \Phi)$ the set of all commodities.

We consider a time slotted system with slots normalized to integral units $t\in\{0,1,2,\dots\}$.
We denote by $A^{(c,\phi)}(t)$ the number of exogenous arrivals of commodity-$(c,\phi)$ packets at node $s_{c}$ during time slot $t$, and by $\lambda^{(c,\phi)}$ its expected value, referred to as the average arrival rate, where we assume that $A^{(c,\phi)}(t)$ is independently and identically distributed (i.i.d.) across time slots. The vector $\boldsymbol{\lambda} = \{\lambda^{(c,\phi)}, (c,\phi) \in (\CC, \Phi)\}$ characterizes the arrival rates to the network.

\section{Policy space and capacity region}
\label{sc:capacity}


We address the {\em mixed-cast service chain control problem}, where both unicast and multicast packets must be processed by a specified chain of service functions before being delivered to their associated destinations.
The goal is to develop a control policy that maximizes network throughput under both communication and computation constraints. 


We first transform the original problem that has both communication and computation constraints into a network flow problem in a graph that only has link capacity constraints. The transformation simplifies the representation of a flow.
We then limit the routing policy space without reducing the capacity region.
Finally, we characterize the network capacity region.

\subsection{Transformation to a layered graph}

Following the approach of \cite{Cao2017}, we model the flow of packets through a service chain via a layered graph, with one layer per stage of the service chain.
Let $\GG^{(\phi)} = (\GG^{(\phi, 0)},\dots,\GG^{(\phi, M_\phi)})$, with edge set $\EE^{(\phi)}$ and vertex set $\VV^{(\phi)}$,
denote the layered graph associated with service chain $\phi$.
Each layer $\GG^{(\phi, i)}$ is an exact copy of 
the original graph $\GG$, used to represent the routing of packets at stage $i$ of service $\phi$, \ie the routing of packets that have been processed by the first $i$ functions of service $\phi$.
Let $u^{(\phi, i)}$ 
denote the copy of node $u$ in $\GG^{(\phi, i)}$, and edge $(u^{(\phi, i)}, v^{(\phi, i)})$ 
the copy of link $(u,v)$ in $\GG^{(\phi, i)}$.
Across adjacent layers, a directed edge from $u^{(\phi, i-1)}$ to $u^{(\phi, i)}$ for all $u \in \NN_{(\phi,i)}$ is used to represent the computation of function $(\phi,i)$. 
See Fig.~\ref{fig:layered} for an example of the layered graph.


\begin{Proposition} \label{th:mapping}
  There is a one-to-one mapping between a flow from $s^{(\phi, 0)}$ to $\DD^{(\phi, M_\phi)}$ in $\GG^{(\phi)}$ and a flow from $s$ to $\DD$ processed by $\phi$ in $\GG$.
\end{Proposition}
\begin{proof}
Let a flow be processed by function $(\phi,i)$ at node $u \in \NN_{(\phi,i)}\subseteq\VV$. Then, by construction of the layered graph, an equivalent flow must traverse link $(u^{(\phi, i-1)},u^{(\phi, i)})\in\EE^{(\phi)}$.
Similarly, let a flow that has been processed by the first $i$ functions of service $\phi$ traverse link $(u,v)\in\EE$. Then, an equivalent flow must traverse link $(u^{(\phi, i)},v^{(\phi, i)})\in \EE^{(\phi)}$.
Under this mapping, every flow processed by $\phi$ in $\GG$ corresponds to a flow in $\GG^{(\phi)}$, and vice versa.
\end{proof}

We now state generalized flow conservations laws in the layered graph that readily apply to the original graph by Proposition \ref{th:mapping}.

Let $f_{u^{(\phi,i)} v^{(\phi,i)}}$ denote the flow rate on link $(u^{(\phi,i)}, v^{(\phi,i)})$, \ie the rate of {\em stage-$i$} packets 
on link $(u,v)$, where a stage-$i$ packet is a packet that has been processed by the first $i$ functions in $\phi$, and not 
by functions $(\phi,i+1), \dots, (\phi, M_\phi)$.
Similarly, $f_{u^{(\phi,i-1)} u^{(\phi,i)}}$ denotes the flow rate on link $(u^{(\phi, i-1)}, u^{(\phi, i)})$, \ie the computation rate at node $u$ for processing stage-$(i-1)$ packets into stage-$i$ packets via function $(\phi,i)$.

We first focus on unicast traffic, where no packet duplication is required.\footnote{Packet duplication is different from flow scaling. Flow scaling is a result of service function processing. An expanded flow, which is a function output, contains different packets. Packet duplication makes identical copies of a packet, which may be forwarded along different routes to reach different destinations for multicast.}
Note that due to non-unit computation requirements and flow scalings, traditional flow conservation does not hold even for unicast traffic.
For a given node $u^{(\phi, i)} \in \GG^{(\phi, i)}$, the following {\em generalized flow conservation} law holds:
\begin{align}\label{flowcons}
   & \sum_{v^{(\phi,i)} \in \VV^{(\phi)}} f_{v^{(\phi,i)} u^{(\phi,i)}} + \frac{\xi^{(\phi,i)}}{r^{(\phi,i)}} f_{u^{(\phi,i-1)} u^{(\phi,i)}} \notag\\
 &= \sum_{v^{(\phi,i)} \in \VV^{(\phi)}} f_{u^{(\phi,i)} v^{(\phi,i)}} + \frac{1}{r^{(\phi,i + 1)}} f_{u^{(\phi,i)} u^{(\phi,i + 1)}} .
\end{align}


\begin{figure}
\centering{\includegraphics[width=\linewidth]{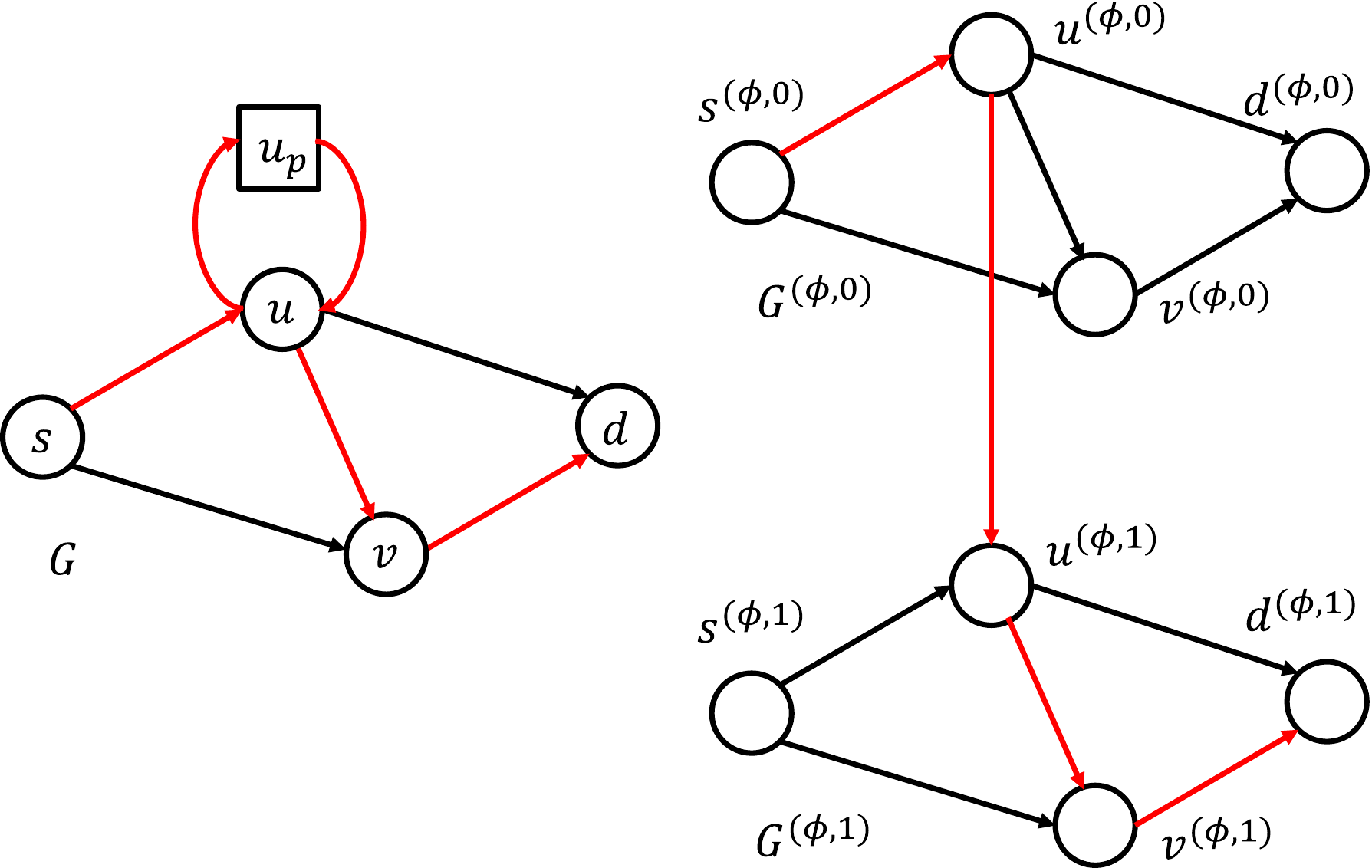}}
\caption{
The left figure is the original graph $\GG$, where $u$ is the only computation node for the single function in $\phi$. A dummy node $u_p$ and connections to $u$ are added to illustrate the availability of service function processing at node $u$. The right figure is the layered graph $\GG^{(\phi)}$.}
\label{fig:layered}
\end{figure}

In the case of multicast traffic, packet duplication is necessary for a packet to reach multiple destinations. Packet duplications can happen at any stage of a service chain. Suppose that a stage-$i$ packet is duplicated. Then, all the copies must be processed by functions $(\phi,i+1),\dots,(\phi,M_\phi)$ before reaching destinations in $\DD$. Equivalently, in the layered graph $\GG^{(\phi)}$, if a packet is duplicated at a node in $\GG^{(\phi, i)}$, then all the copies need to
travel through the links that cross the remaining $M_{\phi}-i$ layers
before reaching a node in $\DD^{(\phi,M_\phi)}$.
The {\em generalized flow conservation and packet duplication} law states that generalized flow conservation \eqref{flowcons} holds at the nodes where there is no packet duplication.

Given the flow rates in the layered graph and the mapping of Proposition \ref{th:mapping}, the flow rates in the original graph can be easily derived.
The communication rate on link $(u,v) \in \GG$, 
computed as the sum over the flow rates on links $(u^{(\phi, i)}, v^{(\phi, i)})$, $\forall \phi \in \Phi, i \in \{0,\dots,M_\phi\}$, and the computation rate at node $u \in \GG$, 
computed as the sum over the flow rates on links $(u^{(\phi, i - 1)}, u^{(\phi, i)})$, $\forall \phi \in \Phi, i \in \{1,\dots,M_\phi\}$ are subject to communication and computation capacity constraints:
\begin{align*}
 & \sum_{\phi \in \Phi, i \in \{0,\dots,M_\phi\}} f_{u^{(\phi,i)} v^{(\phi,i)}} \leq \mu_{uv}, \\
 & \sum_{\phi \in \Phi, i \in \{1,\dots,M_\phi\}} f_{u^{(\phi,i-1)} u^{(\phi,i)}} \leq \mu_{u}.
\end{align*}



\subsection{Policy space} \label{sc:policy}
An admissible policy $\pi$ for the mixed-cast service chain control problem consists of two actions at every time slot $t$.
\begin{enumerate}
  \item {\bf Route selection:} For a commodity-$(c,\phi)$ packet that originates at $s_{c}$ and is destined for $\DD_{c}$, choose a set of links $E^{(c,\phi)} \subseteq \EE^{(\phi)}$, and assign a number of packets\footnote{Recall that a commodity-$(c,\phi)$ input packet can be expanded to multiple packets due to flow scaling and packet duplication.} on each link that satisfies the generalized conservation law for unicast traffic and the generalized conservation and duplication law for multicast traffic.
  \item {\bf Packet scheduling:} Transmit packets through every link in $\EE$ according to a schedule that respects capacity constraints.
\end{enumerate}

The set of all admissible policies is denoted by $\Pi$. The set $\Pi$ includes policies that may use past and future arrival and control information.

Let $\mathcal P^{(c,\phi),\pi}(t)$ denote the packets that are originated at $s_{c}$, processed by $\phi$, and delivered to every node in $\DD_{c}$ under policy $\pi$ up to time $t$. Let $R^{(c,\phi),\pi}(t) = |\mathcal P^{(c,\phi),\pi}(t)|$ denote the number of such packets. The number of packets received by any node in $\DD_{c}$ is at least $\prod_{i=1}^{M_\phi} \xi^{(\phi,i)} R^{(c,\phi),\pi}(t)$ due to flow scaling. We characterize the network throughput using \emph{arrival rates}. A policy $\pi$ \emph{supports} an arrival rate vector $\boldsymbol{\lambda}$ if
\begin{equation}\label{eq:support}
  \liminf_{t \rightarrow \infty} \frac{R^{(c,\phi),\pi}(t)}{t} = \lambda^{(c,\phi)}, ~~ \forall (c,\phi) \in (\CC,\Phi), \text{ w.p.}~1.
\end{equation}

The network layer capacity region is the set of all supportable arrival rates.
\begin{equation}\label{eq:region}
  \Lambda(\GG,\CC,\Phi) = \{\boldsymbol{\lambda} \in \mathbb{R}_+^{|\CC||\Phi|} : \exists \pi \in \Pi \text{ supporting } \boldsymbol{\lambda}\}
\end{equation}



We next restrict the set of admissible routes without reducing the capacity region. 
A route is \emph{efficient} if every packet never visits the same node in $\GG^{(\phi)}$ more than once. For example, if there is no flow scaling, a unicast packet is transmitted through a path from the source to the destination, without cycles, and a multicast packet is transmitted and duplicated through a tree that connects the source and the set of destinations. It suffices to consider efficient routes, by Lemma \ref{th:efficient}, whose proof is in Appendix \ref{sc:ap_restrict}.

\begin{lem} \label{th:efficient}
  Any arrival rate $\boldsymbol{\lambda}$ in the capacity region can be supported by a policy that only uses efficient routes.
\end{lem}

Moreover, we further restrict the route of a unicast packet to be a \emph{service chain path}, and the route of a multicast packet to be a \emph{service chain Steiner tree}, without reducing the capacity region by proving Theorem \ref{th:restrict}. Note that under flow scaling, one commodity-$(c,\phi)$ packet that originates at $s_c$ is scaled to $\prod_{j=1}^{i-1} \xi^{(\phi,j)}$ packets at stage-$(i-1)$. To process them, 
function $(\phi,i)$ requires
$x^{(\phi,i)} = r^{(\phi,i)} \prod_{j=1}^{i-1} \xi^{(\phi,j)}$  computation resource units, and outputs
$w^{(\phi,i)} = \prod_{j = 1}^{i} \xi^{(\phi,j)}$ packets. Let $w^{(\phi,0)} = \xi^{(\phi,0)} = 1$. We define a service chain path as follows.

\begin{defn} \label{def:path}
A commodity-$(c,\phi)$ unicast packet is routed over a \emph{service chain path} $T^{(c,\phi)}$, if
\begin{enumerate}
  \item $T^{(c,\phi)}$ is a path from $s_{c}^{(\phi,0)}$ to $d_{c}^{(\phi, M_\phi)}$ in $\GG^{(\phi)}$;
  \item $w^{(\phi,i)}$ packets are routed over a link in $T^{(c,\phi)}$ that belongs to $\GG^{(\phi,i)}$;
  \item $x^{(\phi,i)}$ packets are routed over a link in $T^{(c,\phi)}$ that connects $\GG^{(\phi,i-1)}$ and $\GG^{(\phi,i)}$.
\end{enumerate}
\end{defn}
It is easy to verify that the generalized flow conservation law holds in a service chain path. Clearly, a service chain path is an efficient route, since every node in $\GG^{(\phi)}$ is visited only once by the same packet. However, an efficient route does not have to be a service chain path. If a packet is expanded into two packets via intermediate service processing, the two packets can take different paths without violating route efficiency. For example, in Fig.~\ref{fig:path}, the left figure illustrates a service chain path, while the right figure illustrates an efficient route that is not a service chain path.
\begin{figure}
\centering{\includegraphics[width=\linewidth]{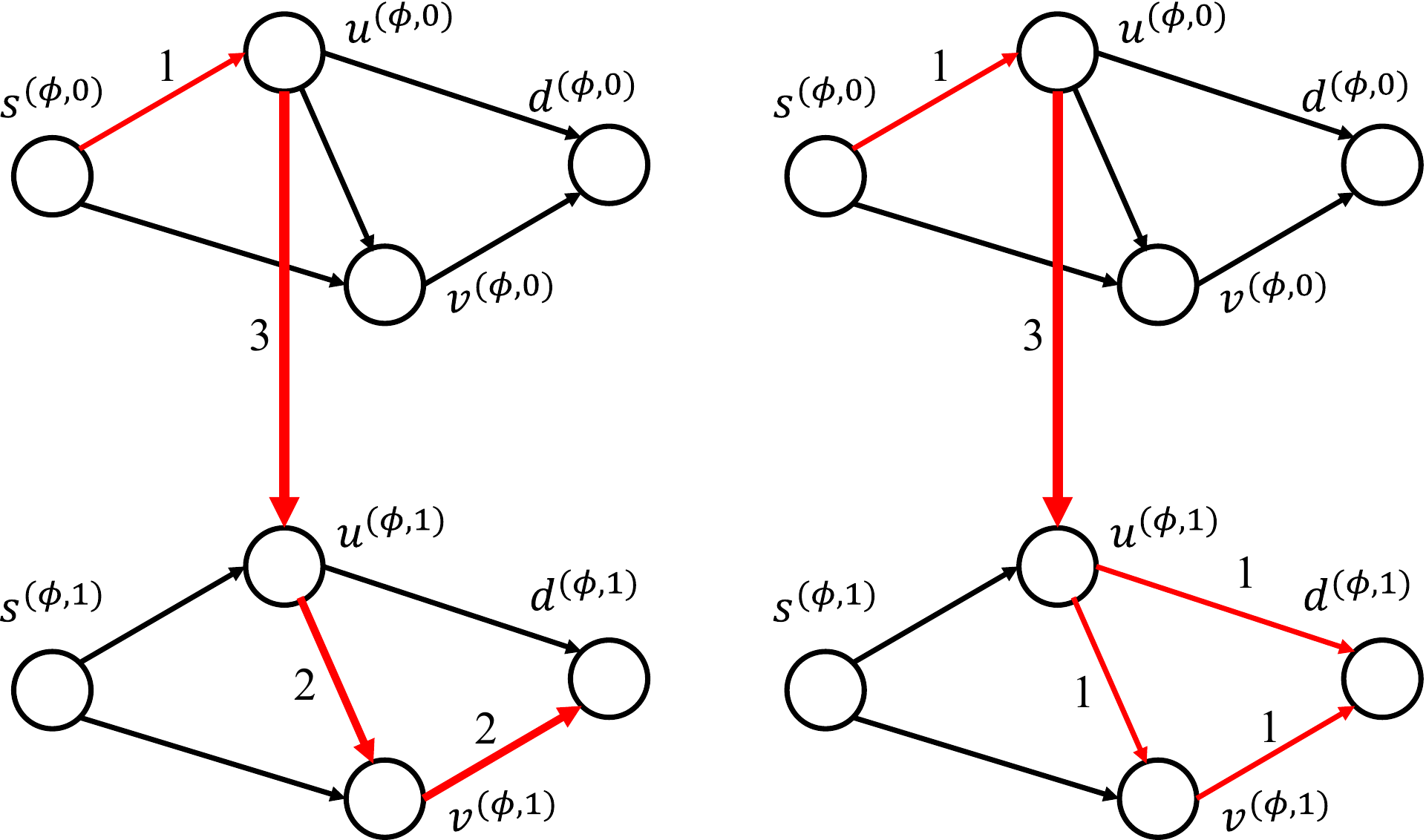}}
\caption{The left figure illustrates a service chain path, and the right figure illustrates an alternative efficient route that is not a service chain path. The number adjacent to a link indicates the number of packets on the link. Scaling factors: $x^{(\phi,1)} = 3$; $w^{(\phi,1)} = 2$.}
\label{fig:path}
\end{figure}

We next define a service chain Steiner tree.
\begin{defn}
A commodity-$(c,\phi)$ multicast packet is routed over a \emph{service chain Steiner tree} $T^{(c,\phi)}$, if
\begin{enumerate}
  \item $T^{(c,\phi)}$ is a Steiner tree (arborescence) that is rooted at $s_{c}^{(\phi,0)}$ and connected to $\DD_{c}^{(\phi, M_\phi)}$ in $\GG^{(\phi)}$;
  \item $w^{(\phi,i)}$ packets are routed over a link in $T^{(c,\phi)}$ that belongs to $\GG^{(\phi,i)}$;
  \item $x^{(\phi,i)}$ packets are routed over a link in $T^{(c,\phi)}$ that connects $\GG^{(\phi,i-1)}$ and $\GG^{(\phi,i)}$.
\end{enumerate}
\end{defn}

If a packet is routed over a service chain Steiner tree $T^{(c,\phi)}$, then packet duplications occur at every node that has more than one outgoing edge in $T^{(c,\phi)}$. The number of packet duplications at a node equals its number of outgoing edges in $T^{(c,\phi)}$ minus one. The generalized flow conservation holds at all other nodes.

We conclude this section with Theorem \ref{th:restrict}, whose proof is in Appendix \ref{sc:ap_restrict}.
\begin{thm} \label{th:restrict}
  There exists a policy that chooses a convex combination of service chain paths for each incoming unicast packet, and a convex combination of service chain Steiner trees for each incoming multicast packet, to support any arrival rate $\boldsymbol{\lambda}$ in the capacity region.
\end{thm}

Due to Theorem \ref{th:restrict}, in the following, we restrict our attention to routing policies that use linear combination of service chain paths or service chain Steiner trees to route incoming packets, without reducing network throughput.

\subsection{Capacity region}
For any arrival rate $\boldsymbol{\lambda} \in \Lambda(\GG,\CC,\phi)$, there exists an admissible policy $\pi$ that takes restricted routes and supports $\boldsymbol{\lambda}$.
Let $\TT^{(c,\phi)}$ denote the set of all service chain paths (or Steiner trees) for commodity-$(c,\phi)$ packets. By taking the time average over the actions of $\pi$, for each commodity $(c,\phi)$, there exists a randomized flow decomposition and routing on $\TT^{(c,\phi)}$. Let $\lambda_k^{(c,\phi)}$ be the average (arrival) flow rate of commodity-$(c,\phi)$ packets over $T_k^{(c,\phi)} \in \TT^{(c,\phi)}$.
\begin{eqnarray} \label{eq:decomp}
  \lambda^{(c,\phi)} = \sum_{T_k^{(c,\phi)} \in \TT^{(c,\phi)}} \lambda_k^{(c,\phi)}, ~~ \forall (c,\phi) \in (\CC, \Phi).
\end{eqnarray}

Moreover, flows should satisfy communication and computation capacity constraints. Commodity-$(c,\phi)$ flow contributes a rate $w^{(\phi,i)} \lambda_k^{(c,\phi)}$ on communication link $(u,v)$ if $(u^{(\phi,i)}, v^{(\phi,i)}) \in T_k^{(c,\phi)}$, and a rate of $x^{(\phi,i)} \lambda_k^{(c,\phi)}$ on computation node $u$ if $(u^{(\phi,i-1)}, u^{(\phi,i)}) \in T_k^{(c,\phi)}$.
Let $\mathcal S_{uv} = \{(k,i,c,\phi): (u^{(\phi,i)},v^{(\phi,i)}) \in T_k^{(c,\phi)}, T_k^{(c,\phi)} \in \TT^{(c,\phi)}, i \in \{0,\dots,M_\phi\}, (c,\phi) \in (\CC, \Phi)\}$
denote the set of commodities that use link $(u,v)$ for transmission.
Let $\mathcal S_u = \{(k,i,c,\phi): (u^{(\phi,i - 1)},u^{(\phi,i)}) \in T_k^{(c,\phi)}, T_k^{(c,\phi)} \in \TT^{(c,\phi)}, i \in \{1,\dots,M_\phi\}, (c,\phi) \in (\CC, \Phi)\}$
denote the set of commodities that use node $u$ for processing.
The communication and computation capacity constraints are
represented by (\ref{eq:comm}) and (\ref{eq:comp}), respectively.

\begin{eqnarray}
  \sum_{(k,i,c,\phi) \in \SSS_{uv}} \hspace{-2mm} w^{(\phi,i)}\lambda_k^{(c,\phi)} \leq \mu_{uv}, ~~ \forall (u,v) \in \EE, \label{eq:comm}\\
  \sum_{(k,i,c,\phi) \in \SSS_u} \hspace{-2mm} x^{(\phi,i)}\lambda_k^{(c,\phi)} \leq \mu_{u}, ~~~~~~~~~ \forall u \in \VV. \label{eq:comp}
\end{eqnarray}

To conclude, the capacity region is characterized by the arrival rates $\boldsymbol{\lambda} = \{\lambda^{(c,\phi)}:(c,\phi) \in (\CC, \Phi)\}$ that satisfy constraints (\ref{eq:decomp}), (\ref{eq:comm}), and (\ref{eq:comp}).

\section{Dynamic routing in a virtual system} \label{sc:routing}

In this section, we study a virtual queueing system for a distributed computing network, 
whose
simplified dynamics allows us to develop a dynamic routing algorithm that guarantees that the average arrival rate at a virtual link is no more than its service rate.
We then formalize the connection between the virtual and physical systems in Section \ref{sc:physical}.

We consider a virtual queueing system $\{\tilde{Q}_{uv}(t), \forall (u,v) \in \EE\}$ and $\{\tilde{Q}_{u}(t), \forall u \in \VV\}$ for network $\GG$.
In contrast to the physical system, in which packets travel through the links in its route sequentially, in the virtual system, a packet immediately enters the virtual queues of all the links in its route, upon arrival at the network.
The number of packets that arrive at the communication queue $\tilde{Q}_{uv}$ at time $t$, denoted by $A_{uv}(t)$, is the sum of the number of packets routed on $(u^{(\phi, i)}, v^{(\phi, i)})$, $\forall \phi \in \Phi, i \in \{0,\dots,M_\phi\}$ at time $t$. Similarly, the number of packets $A_{u}(t)$ that arrive at the computation queue $\tilde{Q}_{u}$ at time $t$ is the sum of the number of packets routed on $(u^{(\phi, i-1)}, u^{(\phi, i)})$, $\forall \phi \in \Phi, i \in \{1,\dots,M_\phi\}$ at time $t$. The value $A_{uv}(t)$ indicates the total number of packets that \emph{will be transmitted} through link $(u,v)$, in order to serve the packets (and their associated packets after processing) that arrive at time $t$, based on the routing decision. The value $A_{u}(t)$ indicates the total amount of computation at node $u$ needed to process these packets.
The departure rate of the packets in $\tilde{Q}_{uv}$ is equal to the transmission capacity of link $(u,v)$, $\mu_{uv}$, and the departure rate of the packets in $\tilde{Q}_u$ is equal to the processing capacity of node $u$, $\mu_u$.


We study the queueing dynamics under a policy that routes all the packets that belong to the same commodity and arrive at the same time, through a service chain path or service chain Steiner tree.
Let $A^{(c, \phi)}(t)$ be the number of commodity-$(c,\phi)$ packets that arrive at the network at time $t$. Let $T^{(c, \phi),\pi}$ denote the path or tree chosen under policy $\pi$ at time $t$. Let $A_{uv}^{(c,\phi),\pi}(t)$ denote the number of packets that arrive at the virtual communication queue $(u,v)$ at time $t$. Recall that $w^{(\phi,i)}$ and $x^{(\phi,i)}$ were defined before Definition \ref{def:path} in Section \ref{sc:capacity}.
\begin{equation} \label{eq:scaling1}
A_{uv}^{(c, \phi),\pi}(t) = \sum_{(u^{(\phi, i)}, v^{(\phi, i)}) \in T^{(c, \phi),\pi}} w^{(\phi, i)} A^{(c, \phi)}(t).
\end{equation}
Let $A_{u}^{(c, \phi), \pi}(t)$ denote the number of packets that arrive at the virtual computation queue at $u$ at time $t$.
\begin{equation} \label{eq:scaling2}
A_{u}^{(c, \phi),\pi}(t) = \sum_{(u^{(\phi, i-1)}, u^{(\phi, i)}) \in T^{(c, \phi),\pi}} x^{(\phi, i)} A^{(c, \phi)}(t).
\end{equation}


The virtual queue lengths $\tilde Q_{uv}(t)$ and $\tilde Q_u(t)$ evolve according to the following recursion, where $(a)^+ = \max(a, 0)$.
\begin{eqnarray*}\label{eq:lindley}
  \tilde Q_{uv}(t + 1) &=& \Big(\tilde Q_{uv}(t) + \sum_{(c,\phi)\in(\CC,\Phi)} A_{uv}^{(c, \phi),\pi}(t) - \mu_{uv} \Big)^+, \\
  \tilde Q_{u}(t + 1) &=& \Big(\tilde Q_{u}(t) + \sum_{(c,\phi)\in(\CC,\Phi)} A_{u}^{(c, \phi),\pi}(t) - \mu_{u} \Big)^+.
\end{eqnarray*}

Based on the virtual queueing system, we study the following
\emph{dynamic routing policy $\pi^*$.}
When $A^{(c, \phi)}(t)$ packets arrive at time $t$, policy $\pi^*$ chooses a route $T^{(c, \phi),\pi^*}$ by minimizing
\begin{eqnarray}
  &&\sum_{(u,v)\in \EE} \tilde{Q}_{uv}(t)A_{uv}^{(c,\phi),\pi}(t) + \sum_{u\in \VV} \tilde{Q}_{u}(t)A_{u}^{(c,\phi),\pi}(t) \nonumber \\
  &=& \hspace{-3mm}A^{(c, \phi)}(t) \Big( \hspace{-9mm}\sum_{(u^{(\phi,i)}, v^{(\phi,i)}) \in \EE^{(\phi)}} \hspace{-8mm} w^{(\phi,i)} \tilde{Q}_{uv}(t) 1\{(u^{(\phi,i)}, v^{(\phi,i)}) \in T^{(c, \phi),\pi}\} \nonumber \\
  && + \hspace{-10mm} \sum_{(u^{(\phi,i-1)},u^{(\phi,i)}) \in \EE^{(\phi)}} \hspace{-10mm} x^{(\phi,i)} \tilde{Q}_u(t) 1\{(u^{(\phi,i-1)}, u^{(\phi,i)}) \in T^{(c, \phi),\pi}\} \Big).\label{eq:mintree}
\end{eqnarray}

\emph{Computation of $\pi^*$: }
Policy $\pi^*$ can be computed by applying standard graph algorithms to the layered graph $\GG^{(\phi)}$. Let the length of link $(u^{(\phi,i)}, v^{(\phi,i)})$ be $ w^{(\phi,i)} \tilde{Q}_{uv}(t)$, and the length of link $(u^{(\phi,i-1)},u^{(\phi,i)})$ be $x^{(\phi,i)} \tilde{Q}_u(t)$. For unicast traffic, the optimal path is the \emph{shortest path} from $s^{(\phi,0)}$ to $d^{(\phi,M_\phi)}$. For multicast traffic, the optimal tree is the \emph{minimum Steiner tree} from $s^{(\phi,0)}$ to $\DD^{(\phi,M_\phi)}$.

\emph{Interpretation of $\pi^*$: }
Policy $\pi^*$ avoids routing packets through overloaded links and nodes by assigning higher costs to those with larger virtual queue lengths. Moreover, the scaling of packets are accounted for the additional load that an incoming packet contributes to the links and nodes. Both load and scaling factors are captured in the design of link cost function: the length (cost) of a link is the product of the scaling factor and the virtual queue length. Policy $\pi^*$ determines the routes in the original graph $\GG$, in addition to the routes in the layered graph $\GG^{(\phi)}$, due to Proposition \ref{th:mapping}.

\emph{Performance of $\pi^*$: }
Policy $\pi^*$ stabilizes the virtual system for any arrival rate in the interior of the capacity region.
\begin{thm} \label{th:strong}
  Under routing policy $\pi^*$, the virtual queue process $\{\tilde{Q}(t)\}_{t \geq 0}$ is strongly stable for any arrival rate that is in the interior of the capacity region. \Ie
  $$\limsup_{T \rightarrow \infty} \frac{1}{T}\sum_{t = 0}^{T - 1}\Big(\sum_{(u,v) \in \EE} \E \tilde{Q}_{uv}(t) + \sum_{u \in \VV} \E \tilde{Q}_{u}(t)\Big) < \infty.$$
\end{thm}

The proof of Theorem \ref{th:strong} is based on Lyapunov drift analysis and can be found in Appendix \ref{sc:ap_virtual}. The queue stability implies that the arrival rate at each virtual queue is no more than its service rate.
\section{Control of the physical network}
\label{sc:physical}
In this section, we formalize the connection between the virtual system and the physical system, and develop a throughput-optimal control policy for a distributed computing network. 
Recall that an admissible policy consists of two actions at every time slot: 1) route selection, 2) packet scheduling.

The route selection for an incoming packet to the network is identical to the route selection $\pi^*$ in the virtual system. Suppose that a packet is served (\ie both processed by all the service functions and delivered to the destination) by the network. The amount of traffic that the packet contributes to a physical queue $Q_{uv}$ (or $Q_u$) is the same as the amount of traffic that it contributes to the virtual queue $\tilde Q_{uv}$ (or $\tilde Q_u$). 
Strong stability of virtual queues implies that the average arrival rate is at most the service rate of each virtual queue under $\pi^*$. Therefore, by applying the same routing policy to the physical system, the average arrival rate (or offered load) is at most the service rate for each physical queue. The statement is made precise in the proof of Theorem \ref{th:rate}.

A packet scheduling policy chooses a packet to transmit over a link or to process at a node, when there is more than one packet awaiting service. It was proved in \cite{umw_info17, nto} that an extended nearest-to-origin (ENTO) policy guarantees queue stability, as long as the average arrival rate is no more than the service rate at each queue. The ENTO policy gives higher priority to packets that have traveled a smaller number of hops (\ie closer to their origins). A duplicated packet (for multicast) inherits the hop count of the original packet. In the proof of Theorem \ref{th:rate}, we show that this policy guarantees the stability of physical queues 
even with flow scaling (\ie one packet processed by a first queue may enter a second queue in the form of multiple packets).

The resulting routing and scheduling policy, referred to as Universal Computing Network Control (UCNC), is summarized in Algorithm \ref{alg:umws}.

\begin{algorithm}[h]
\caption{Universal Computing Network Control (UCNC).}
Initialization: $\tilde Q_{uv}(0) = \tilde Q_u(0) = 0$, $\forall (u,v) \in \EE, u \in \VV$. \\

At each time slot $t$:
\begin{enumerate}
\item {\bf Preprocessing}. For an incoming commodity-$(c,\phi)$ packet, construct a layered graph $\GG^{(\phi)}$. Let the cost of link $(u^{(\phi,i)}, v^{(\phi,i)})$ be $ w^{(\phi,i)} \tilde{Q}_{uv}(t)$, and the cost of link $(u^{(\phi,i-1)},u^{(\phi,i)})$ be $x^{(\phi,i)} \tilde{Q}_u(t)$.
\item {\bf Route Selection ($\pi^*$)}. Compute a minimum-cost route $T^{(c,\phi),\pi^*}$ for a commodity-$(c,\phi)$ incoming packet. The packet will follow $T^{(c,\phi),\pi^*}$ for transmission and processing. 
\item {\bf Packet Scheduling (ENTO)}. Each physical link transmits packets and each computation node processes packets according to the ENTO policy.
\item {\bf Virtual Queues Update}.
\begin{eqnarray*}
  \tilde Q_{uv}(t + 1) &=& \Big(\tilde Q_{uv}(t) + \hspace{-3mm} \sum_{(c,\phi)\in(\CC,\Phi)} \hspace{-5mm} A_{uv}^{(c, \phi),\pi^*}(t) - \mu_{uv}\Big)^+; \\
  \tilde Q_{u}(t + 1) &=& \Big(\tilde Q_{u}(t) + \hspace{-3mm} \sum_{(c,\phi)\in(\CC,\Phi)} \hspace{-5mm} A_{u}^{(c, \phi),\pi^*}(t) - \mu_{u}\Big)^+.
\end{eqnarray*}
\end{enumerate}
\label{alg:umws}
\end{algorithm}

In Step 2, a commodity-$(c,\phi)$ packet enters the physical network and will be transmitted and processed in $\GG$ according to $T^{(c,\phi),\pi^*} \subseteq \GG^{(\phi)}$ by the mapping in Proposition \ref{th:mapping}. To implement the algorithm, the packet stores $T^{(c,\phi),\pi^*}$. At time slot $t' \geq t$, if it has been processed by the first $i$ functions and is at node $u$, then it enters the physical queue for link $(u,v)$ if $(u^{(\phi,i)}, v^{(\phi,i)}) \in T^{(c,\phi),\pi^*}$. It enters the computation queue at node $u$ if $(u^{(\phi,i)}, u^{(\phi,i+1)}) \in T^{(c,\phi),\pi^*}$. The packet is duplicated (for multicast) if $u^{(\phi,i)}$ has more than one outgoing edge in $T^{(c,\phi),\pi^*}$.

\begin{thm} \label{th:rate}
  Under UCNC, all physical queues are rate stable for any arrival rate in the interior of the capacity region. \Ie
  \begin{eqnarray*}
    \lim_{t \rightarrow \infty}\frac{Q_{uv}(t)}{t} &=& 0, ~~\text{w.p.}~1, ~\forall (u,v) \in \EE; \\
    \lim_{t \rightarrow \infty}\frac{Q_{u}(t)}{t} &=& 0, ~~\text{w.p.}~1, ~\forall u \in \VV.
  \end{eqnarray*}
\end{thm}

The proof can be found in Appendix \ref{sc:ap_phy} and consists of two parts. The first part is to prove that the average arrival rate is no more than the service rate of every link and every computation node. The second part is to prove that under this condition, the physical queues are stable under the ENTO policy. Using standard queue stability analysis (\eg \cite{umw_info17}), we conclude that the policy is throughput-optimal.

\section{Simulation results}
\label{sc:simulation}
In this section, we evaluate the performance of UCNC in a distributed computing network based on the Abilene network topology in Fig. \ref{fig:us}.
For simplicity, we assume that each link is bidirectional and has unit transmission capacity in each direction. We evaluate the performance of UCNC for unicast traffic in Section \ref{sc:unicast}, and for multicast traffic in Section \ref{sc:multicast}. 
In Sections \ref{sc:unicast} and \ref{sc:multicast}, we consider a small number of commodities, and assume that nodes 3 and 8 have unit computation capacity and that all the other nodes have zero computation capacity. In Section \ref{sc:large}, we consider a larger number of commodities with a mix of unicast and multicast.

For unicast traffic, we compare UCNC with the backpressure-based algorithm in \cite{dcnc_icc16}. While both algorithms are throughput-optimal, UCNC yields much shorter packet delay. We also compare UCNC with heuristic policies such as choosing the closest server to process the service functions, and observe that the heuristic policies are not always throughput-optimal. This demonstrates the importance of joint optimization of communication and computation resources.

For multicast traffic, we illustrate the performance of UCNC, and compare the capacity region under multicast traffic with the capacity region when multicast flows are treated as multiple unicast flows. Numerical results indicate
the ability to deliver higher rates when multicast traffic can be served via proper packet duplications, as opposed to creating independent copies for each destination.
This confirms the importance of the first throughput-optimal algorithm for multicast traffic in distributed computing networks.

We compare different policies using the average delay metric. Note that we did not claim any theoretical delay guarantee of UCNC (other than $o(t)$ delay with probability 1 due to Little's law and Theorem \ref{th:rate}). Nevertheless, the delay metric is important for quality of service. Moreover, queue lengths can be inferred from delay information. Small delays indicate short queue lengths and therefore stable queues. Thus, the arrival rates supported by different policies can be inferred using delay information.

\subsection{Unicast traffic} \label{sc:unicast}
\subsubsection{Comparison with backpressure-based algorithm}
We consider two commodities of unicast traffic. The first commodity originates at node 1 and is destined for node 11. The second commodity originates at node 4 and is destined for node 7. Packets in both commodities are processed by two functions in a service chain. Let $\lambda_1$ and $\lambda_2$ denote the expected arrival rates of the two commodities, respectively. Ignoring all the scalings ($\xi = r = 1$), the computation resource constraints are tight 
to support $\lambda_1 + \lambda_2 = 1$. Thus, the capacity region is $\lambda_1 + \lambda_2 \leq 1$. Figure \ref{fig:delay} compares the average packet delays under UCNC and the backpressure-based algorithm, for different arrival rates that satisfy $\lambda_1 = \lambda_2$. We observe that the average packet delays under UCNC are significantly lower than the delays under the backpressure-based algorithm.

\subsubsection{Comparison with nearest-to-destination service function placement}
We compare the performance of UCNC with the heuristic of placing the service functions in the computation node that is nearest to the destination. For a fair comparison, the processing capacity of a single node should be sufficient. We consider a single unicast commodity from node 2 to node 7. The service chain $\phi$ has a single function $(\phi, 1)$ with flow scaling factor $\xi^{(\phi,1)} = 1/3$ and computation requirement $r^{(\phi,1)} = 1/3$. The heuristic policy routes the packets from node 2 to node 8, which is the closest computation node to node 7, processes the packets at node 8, and routes the processed packets from node 8 to node 7. The average packet delays under both algorithms are compared in Fig. \ref{fig:heu1}. Due to communication constraints, the maximum rate that UCNC can support is $\lambda = 3$, while the maximum rate that the heuristic policy can support is $\lambda = 2$. The heuristic policy fails to be throughput-optimal when there is flow scaling (shrinkage) due to processing. This demonstrates the importance of jointly optimizing communication and computation resources.

\subsubsection{Comparison with nearest-to-source service function placement}
Placing a service function at the nearest-to-source computation node may decrease the supportable service rate, when there is flow expansion. We consider a single commodity from node 2 to node 7. The service chain $\phi$ has a single function $(\phi, 1)$ with flow scaling factor $\xi^{(\phi,1)} = 3$ and computation requirement $r^{(\phi,1)} = 1$. The heuristic policy routes the packets from node 2 to node 3, which is the closest computation node to the source, processes the packets at node 3, and then routes the processed packets from node 3 to node 7. The maximum flow rate from node 3 to node 7 is two. Thus, the maximum supportable service rate is $\lambda = 2/3$, which expands to a flow of rate two after processing. In contrast, illustrated in Fig. \ref{fig:heu2}, UCNC is able to support a service rate $\lambda = 1$. This, again, demonstrates the need to jointly optimize communication and computation resources.

\begin{figure}
\centering
\includegraphics[width=.7\linewidth]{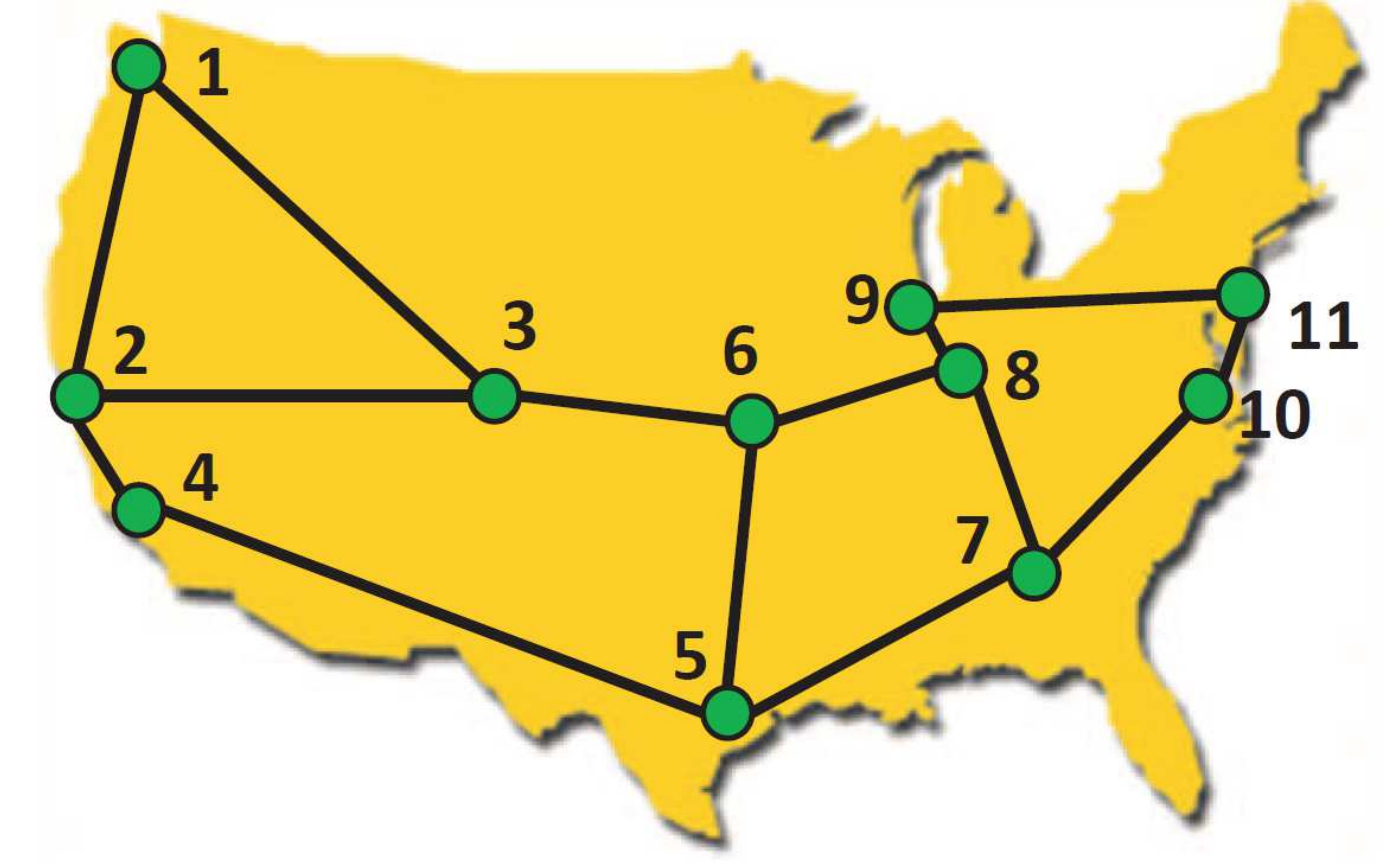}
\caption{Abilene network topology.}
\label{fig:us}
\end{figure}

\begin{figure*}
\centering
\subfigure[]{
\includegraphics[width=.32\linewidth]{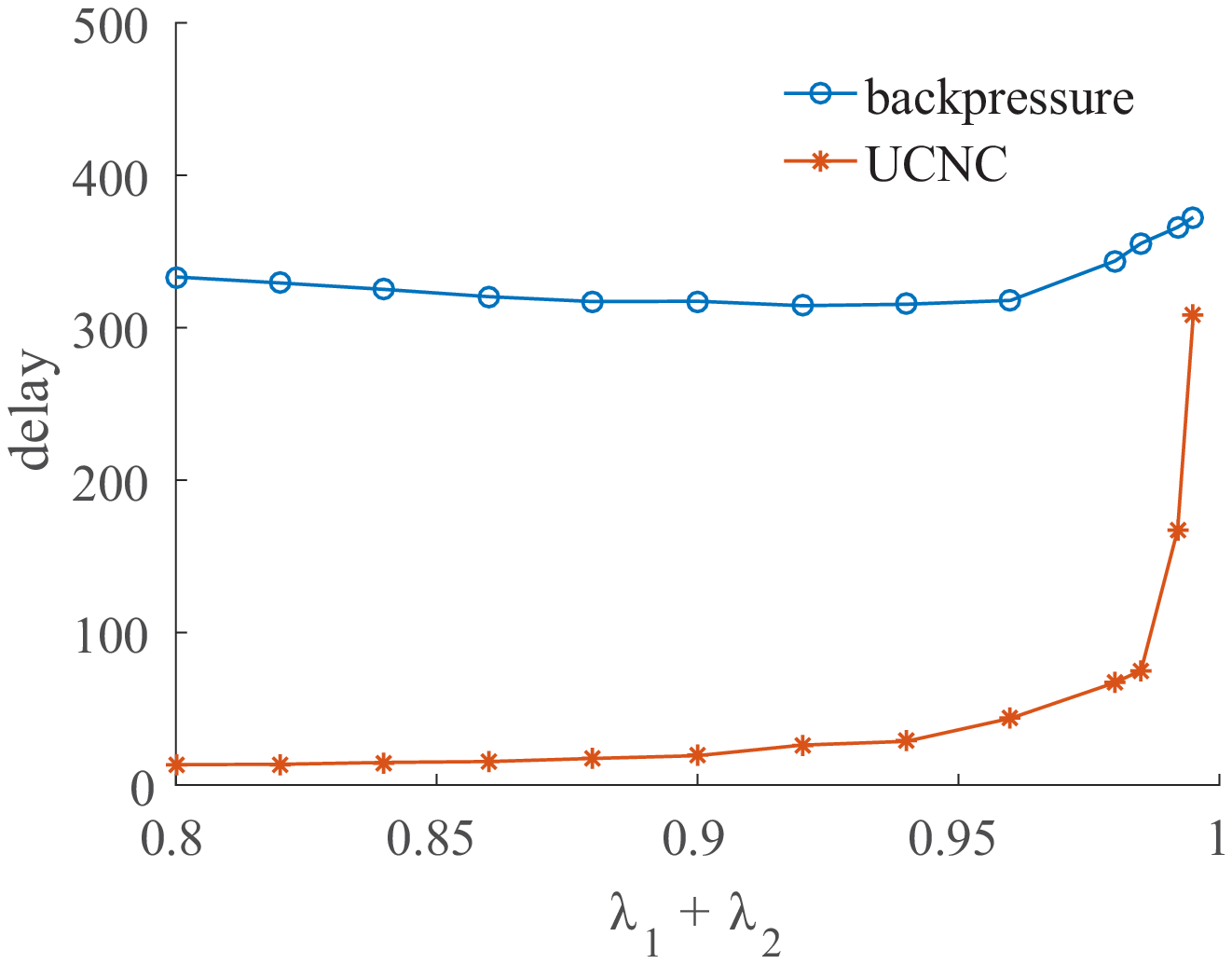}
\label{fig:delay}
\hspace{-.02\linewidth}
}
\subfigure[]{
\includegraphics[width=.32\linewidth]{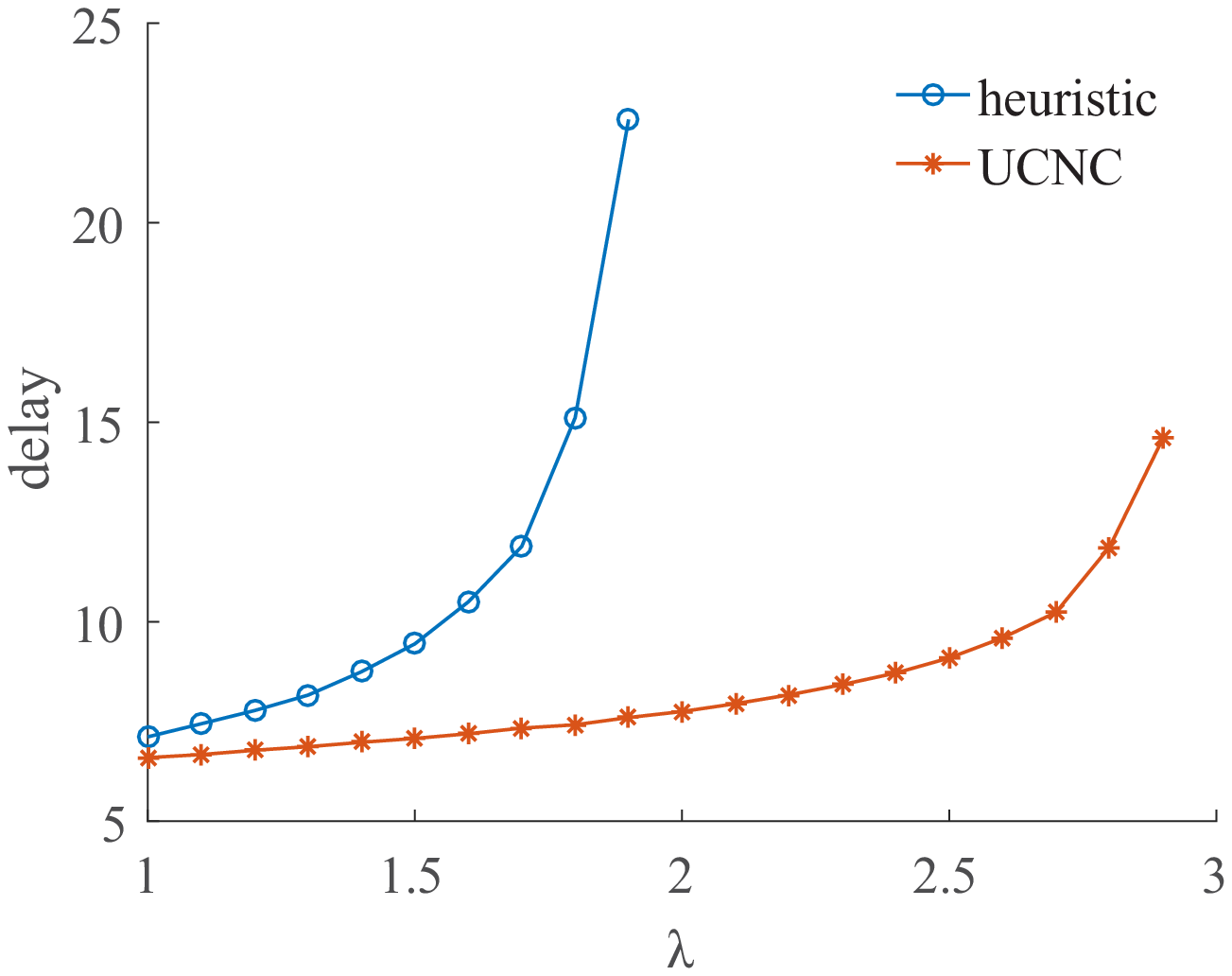}
\label{fig:heu1}
\hspace{-.02\linewidth}
}
\subfigure[]{
\includegraphics[width=.32\linewidth]{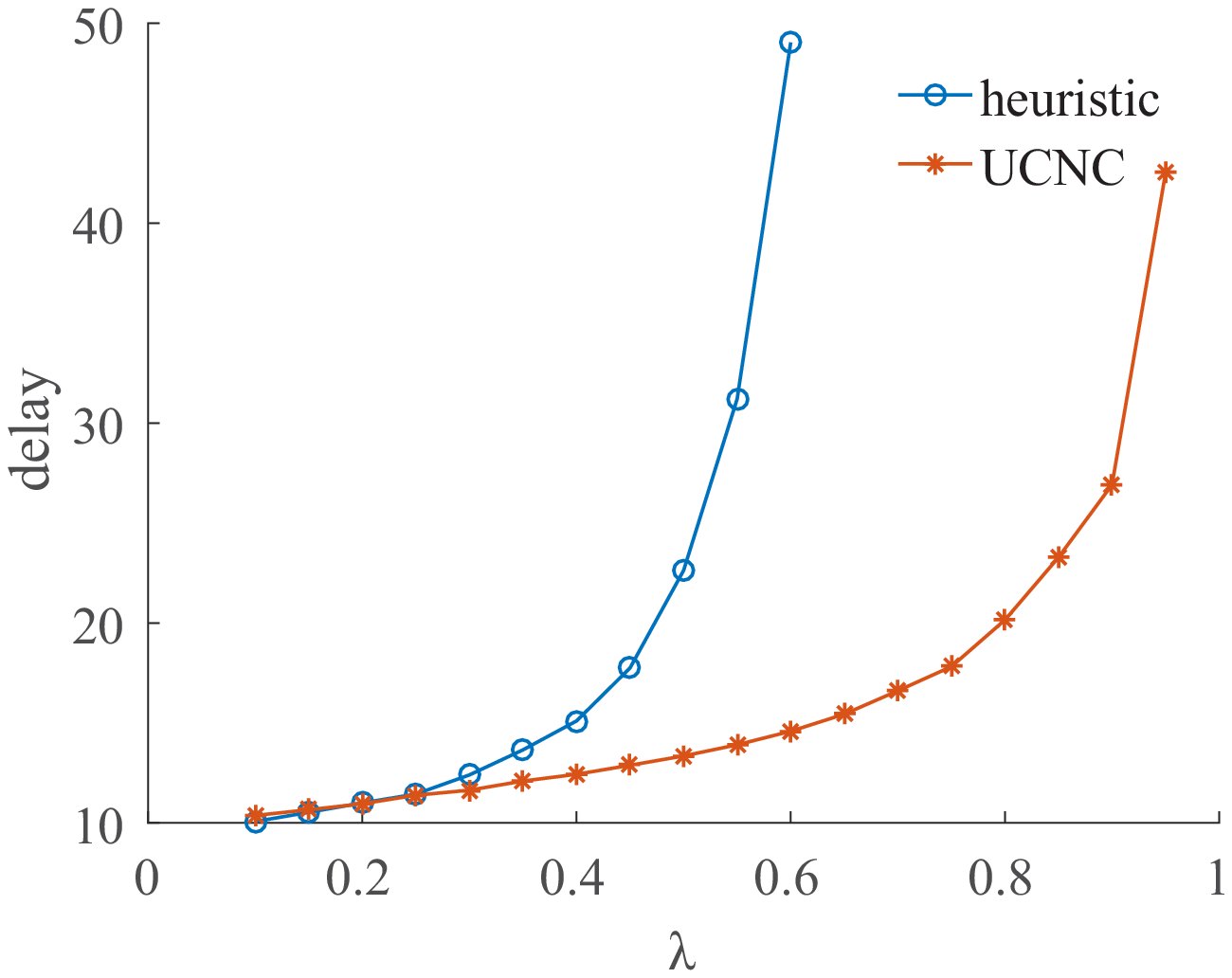}
\label{fig:heu2}
}
\caption{Average packet delay performance: (a) UCNC v.s. backpressure-based algorithm; (b) UCNC v.s. nearest-to-destination function placement heuristic; (c) UCNC v.s. nearest-to-source function placement heuristic.}
\end{figure*}

%
%

\subsection{Multicast traffic} \label{sc:multicast}
We next study a multicast flow from node 1 to nodes 7 and 11. Suppose that the service chain has two functions and that all the scaling factors $\xi, r$ are one. The optimal policy is to process the packets at both nodes 3 and 8, and then duplicate the processed packets and route them to the two destinations. The maximum supportable service rate is $\lambda = 1$ for both destinations. In contrast, if the multicast flow is treated as two unicast flows, then the sum of the service rates to both destinations is one. Thus, multicasting improves the performance of the distributed computing network. As shown in Fig. \ref{fig:mixed}, UCNC is throughput-optimal for multicast traffic, and the average packet delays are small.

\begin{figure}
\centering
\includegraphics[width=.7\linewidth]{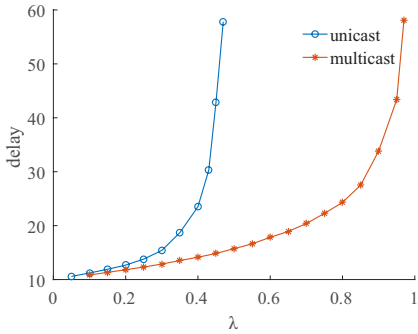}
\caption{Average packet delay of multicast traffic and when multicast is treated as multiple unicast traffic.}
\label{fig:mixed}
\end{figure}

\subsection{Mixed-cast traffic} \label{sc:large}
We evaluate the performance of UCNC under a large number of commodities. We consider three service chains $\Phi = \{\phi_1, \phi_2, \phi_3\}$. Services $\phi_1, \phi_2$ have two functions each, and $\phi_3$ has three functions. The scaling factors $\xi, r$ are chosen independently from a uniform distribution in $[0.5,2]$. Each service chain processes four unicast flows and two multicast flows, where the source and the destination(s) of each flow are randomly chosen among all nodes that are at least two hops away. Thus, there are a total of 18 commodities. Each function can be computed at four randomly chosen computation nodes, each of which has unit capacity.

The average packet delays under the $18$ mixed-cast commodities are shown in Fig. \ref{fig:large}, where all commodities have identical arrival rate $\lambda$. We observe that UCNC is able to support rate $\lambda = 0.12$. In contrast, when each multicast flow is treated as multiple unicast flows, for a total of 24 commodities, the maximum supportable rate is around $\lambda = 0.09$. This demonstrates the importance of optimal control for multicast traffic.
The average packet delays under the backpressure-based algorithm, with multicast flows treated as multiple unicast flows, are over 1000 for $\lambda \in [0.01, 0.09]$, substantially higher than the delays under UCNC, and hence ommitted in the figure.

\begin{figure}
\centering
\includegraphics[width=.7\linewidth]{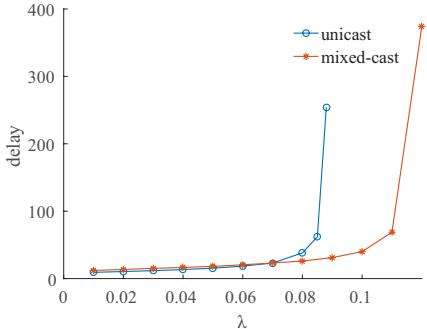}
\caption{Average packet delay of mixed-cast traffic and when multicast is treated as multiple unicast traffic.}
\label{fig:large}
\end{figure}

Finally, we also evaluated the performance of an algorithm that uses the routing policy $\pi^*$ and the First-In-First-Out (FIFO) scheduling policy for the physical queues. Numerical results demonstrate that the average packet delays are close to the delays under the ENTO scheduling policy, and are omitted for brevity. Thus, for practical purpose of dynamic control in distributed computing networks, FIFO scheduling policy could also be used.

\section{Extensions}
\label{sc:extensions}
\subsection{Undirected network}
In an undirected network where the sum of transmission rates in both directions over a link is limited by the link capacity, the virtual queue length updates should be modified while the other steps in UCNC remain the same. Notice that the capacity region of an undirected network is given by the flow decomposition constraints \eqref{eq:decomp}, the computation constraints \eqref{eq:comp}, and the following communication constraints.
$$\sum_{(k,i,c,\phi) \in \SSS'_{uv}} \hspace{-2mm} w^{(\phi,i)}\lambda_k^{(c,\phi)} \leq \mu_{uv}, ~~ \forall (u,v) \in \EE,$$
where $\SSS'_{uv} = \{(k,i,c,\phi): (u^{(\phi,i)},v^{(\phi,i)}) \in T_k^{(c,\phi)} \text{ or } (v^{(\phi,i)},u^{(\phi,i)}) \in T_k^{(c,\phi)}, T_k^{(c,\phi)} \in \TT^{(c,\phi)}, i \in \{0,\dots,M_\phi\}, (c,\phi) \in (\CC, \Phi)\}$ accounts for transmission in both directions of the link.

Each undirected link $(u,v)$ is associated with a virtual queue. A packet contributes load to the queue if it will travel either from $u$ to $v$, or from $v$ to $u$. Eq. \eqref{eq:scaling1} is updated by
\begin{eqnarray*}
  A_{uv}^{(c, \phi),\pi}(t) &=& \sum_{(u^{(\phi, i)}, v^{(\phi, i)}) \in T^{(c, \phi),\pi}} w^{(\phi, i)} A^{(c, \phi)}(t) \\
   &+& \sum_{(v^{(\phi, i)}, u^{(\phi, i)}) \in T^{(c, \phi),\pi}} w^{(\phi, i)} A^{(c, \phi)}(t).
\end{eqnarray*}

With this modified queue evolution, under the routing policy $\pi^*$, which routes a unicast packet over a shortest path and routes a multicast packet over a minimum Steiner tree, all the virtual queues are strongly stable for any arrival rate in the interior of the capacity region. This can be proved in the same approach as the proof for Theorem \ref{th:strong}.

Thus, the sum of the average packet arrival rates to a link in both directions is no more than the transmission capacity of the link. ENTO scheduling policy still guarantees the stability of physical queues when the link is undirected. Thus, the same routing and scheduling policy, with modified queue evolutions, is throughput-optimal.

\subsection{Network throughput under approximate routing}
The computation for the minimum Steiner tree requires exponential time. To reduce the computation overhead, approximation algorithms can be implemented to compute a near-optimal Steiner tree. We study the performance of UCNC under sub-optimal routing.

Consider a routing policy $\pi^{\alpha}$ that computes a Steiner tree for a multicast packet whose cost is at most $\alpha > 1$ times the minimum cost. The routing policy $\pi^{\alpha}$ and the ENTO scheduling policy are able to support arrival rate vector $\boldsymbol{\lambda} / \alpha$, where $\boldsymbol{\lambda}$ is any vector in the interior of the stability region.

The proof follows a similar approach, by comparing the Lyapunov drift under $\pi^{\alpha}$ with the (scaled) drift under a randomized policy that supports arrival rate vector $\boldsymbol{\lambda} / \alpha$.

\subsection{Broadcast and anycast traffic}
Broadcast traffic is a special case of multicast traffic, where the destination nodes of a commodity include all the nodes in $\VV$. At each time $t$, for a commodity-$(c,\phi)$ packet, the routing policy $\pi^*$ computes a minimum Steiner tree that is rooted at $s_c^{(\phi,0)}$ and connected to $\VV^{(\phi, M_\phi)}$.

For anycast traffic, where a commodity-$(c,\phi)$ packet is originated at $s_c$ and destined for any node in $\DD_c$, a dummy node $d_c^{(\phi, M_\phi)}$ is added in $\GG^{(\phi, M_\phi)}$. Links of zero cost are added from $\DD_c^{(\phi, M_\phi)}$ to $d_c^{(\phi, M_\phi)}$. The routing policy $\pi^*$ computes a shortest path from $s_c^{(\phi,0)}$ to $d_c^{(\phi, M_\phi)}$.

Under the new routing policies and the same ENTO scheduling policy, UCNC is throughput-optimal for broadcast and anycast traffic.

\subsection{Location-dependent computation requirements}
UCNC can be extended to handle the problem where a service function $(\phi, i)$ may have different computation resource requirements at different computation nodes. This is motivated by the availability of dedicated servers to process certain functions. The processing of such functions require smaller amount of computation resources at dedicated servers compared with generic servers.

For route selection, the cost of an edge $(u^{(\phi, i-1)}, u^{(\phi, i)})$ at time $t$ is modified to $x_u^{(\phi,i)} \tilde Q_u(t)$, where $\tilde Q_u(t)$ is the virtual queue length of $u$ and $x_u^{(\phi,i)} = r_u^{(\phi,i)} \prod_{j = 1}^{i - 1} \xi^{(\phi,j)}$. The $r_u^{(\phi,i)}$ denotes the computation resource requirement to process each unit of input flow by function $(\phi,i)$ at node $u$, and may depend on $u$.


\section{Conclusion}
\label{sc:conclusion}
We characterized the capacity region and developed the first throughput-optimal control policy (UCNC) for unicast and multicast traffic in a distributed computing network. UCNC handles both communication and computation constraints, flow scaling through service function chains, and packet duplications. Simulation results suggest that UCNC has superior performance compared with existing algorithms.

\section{Appendix}
\label{sc:ap}
\subsection{Restricted routes do not reduce the capacity region}\label{sc:ap_restrict}
\emph{Proof of Lemma \ref{th:efficient}:}
We prove that, any packet that can be transmitted from the source to the destination(s) by time $t$ under a policy $\pi$ that uses arbitrary routes, can also be transmitted from the source to the destination(s) by time $t$ under a policy $\pi'$ that only uses efficient routes. Then, by Eq.~(\ref{eq:support}), any rate $\boldsymbol{\lambda}$ that is supported by $\pi$ can also be supported by $\pi'$. By Eq.~(\ref{eq:region}), any rate in the capacity region can be supported by a policy that only uses efficient routes.

Consider a policy $\pi$ that transmits the same packet to a node in $\GG^{(\phi)}$ more than once. For unicast traffic, where there is no packet duplication, the packet travels through one or more cycles. Moreover, each cycle must be in one layer of $\GG^{(\phi)}$ and the packet can not be processed while traveling through the cycle, since there is no edge from $\GG^{(\phi,j)}$ to $\GG^{(\phi,i)}$ for $i < j$. Construct a policy $\pi'$ by removing all the cycles and transmission schedules 
on the cycle links. Any packet that arrives at a node (\eg the destination) by time $t$ under $\pi$ can also arrive at the same node by time $t$ under $\pi'$.

For multicast traffic, if the same packet visits the same node $u^{(\phi,i)} \in \GG^{(\phi)}$ more than once under policy $\pi$, then there are two possibilities.
\begin{enumerate}
  \item The packet travels through one or more cycles in $\GG^{(\phi,i)}$.
  \item A packet is duplicated at some node $v$, and more than one copy has traveled through some links and reached $u^{(\phi,i)}$.
\end{enumerate}
To construct a policy $\pi'$ that only uses efficient routes, we handle the first case in the same manner as the unicast case. {\it I.e.,} remove all the cycles and the transmission schedules of the packet on the cycles. For the second case, policy $\pi'$ only keeps the routing and scheduling of the packet that first arrives at $u^{(\phi,i)}$, and removes all the duplications that arrive later. If the packet needs to be transmitted through more than one outgoing link from $u^{(\phi,i)}$ under $\pi$, then duplications occur at $u^{(\phi,i)}$ and the duplicated copies follow the same routes and schedules as those under $\pi$.

It is easy to check that the time that a packet visits a node under $\pi'$ is the first time that the packet visits the node under $\pi$. By repeating the process until no packet visits the same node more than once, the policy $\pi'$ only uses efficient routes.

\emph{Remark: }If all scaling factors $w, x$ are one, then an efficient route for a unicast packet is a path from the source to the destination in the layered graph. An efficient route for a multicast packet is a Steiner tree from the source to the destinations in the layered graph.

\emph{Proof of Theorem \ref{th:restrict}:}
If $w^{(\phi,i)} = x^{(\phi,i)} = 1$, $\forall \phi \in \Phi, i \in \{1,\dots,M_\phi\}$, the theorem follows immediately from Lemma \ref{th:efficient}. Next we study arbitrary (rational) scaling factors. We divide a packet into \emph{micro packets}, and represent the routes of a packet by the composition of paths (or Steiner trees) of micro packets.\footnote{The split into micro packets is mostly useful for the analysis of multicast flow and is not necessary for the analysis of unicast flow. However, we choose to use the decomposition for both unicast and multicast flows, for a unified treatment of the two cases.}

\emph{Unicast traffic: }
Consider a policy $\pi'$ that chooses an efficient route $E_p$ for a unicast packet of commodity $(c,\phi)$. We assume that a rational number of packets are routed in each link. A micro packet is designed such that it changes size as it travels through the layered graph, and is never split. Each arriving packet is split into $z$ micro packets. Due to scaling, the size of a micro packet on a link in $\GG^{(\phi,i)}$ is $w^{(\phi,i)}/z$. The size of a micro packet on a link that connects $\GG^{(\phi,i - 1)}$ and $\GG^{(\phi,i)}$ is $x^{(\phi,i)}/z$. The choice of $z$ satisfies the following two constraints.
\begin{enumerate}
\item All the links in $E_p$ carry an integer number of micro packets.
\item Every packet is divided to an integer number of micro packets.\footnote{The second constraint is necessary only for multicast flows.}
\end{enumerate}
Fig.~\ref{fig:micro} illustrates the split into micro packets.

\begin{figure}
\centering
\includegraphics[width=0.6\linewidth]{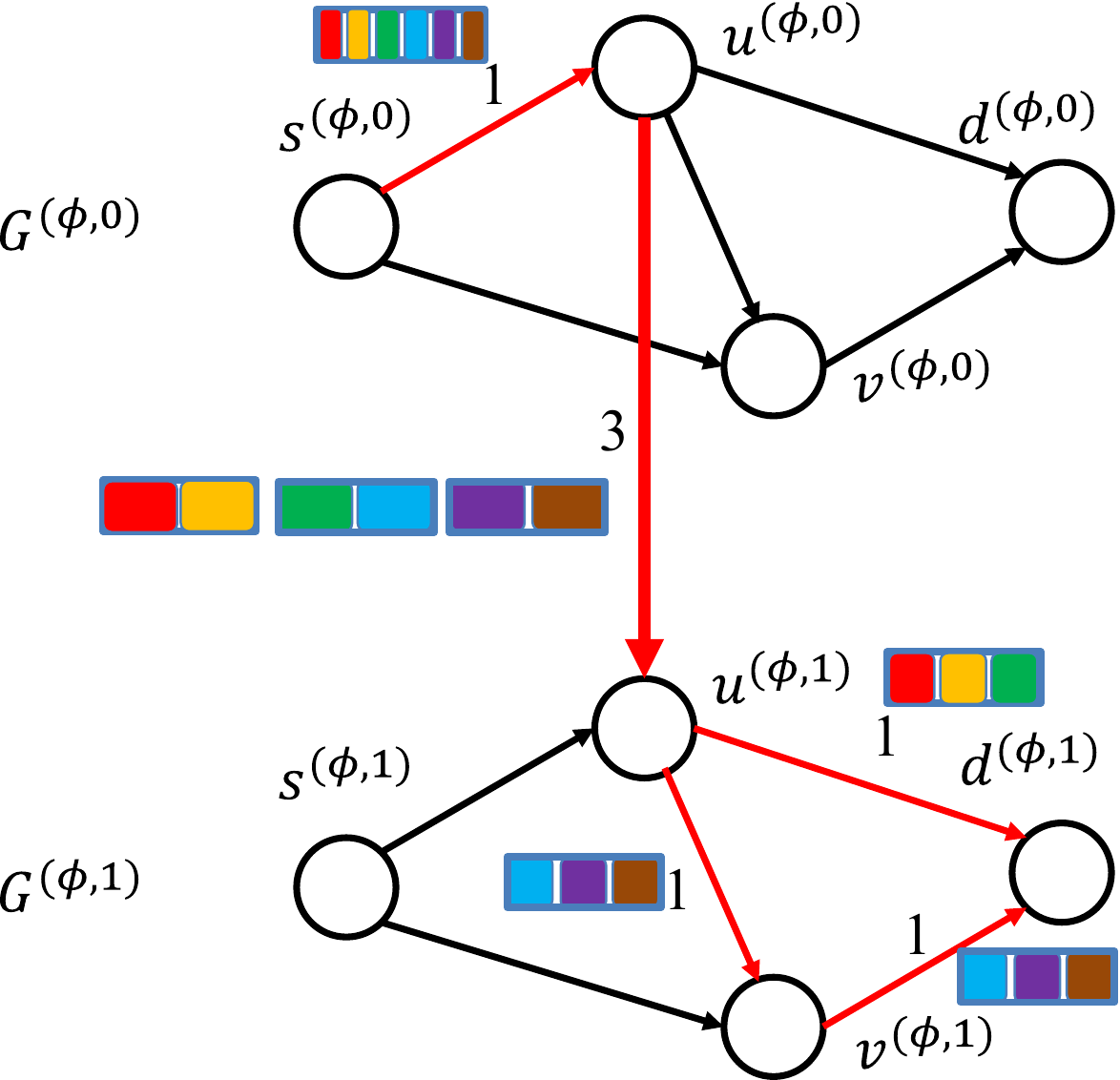}
\caption{Micro packets in an efficient route. The sizes of a micro packet on a link in $\GG^{(\phi,0)}$, link $(u^{(\phi,0)}, u^{(\phi,1)})$, and a link in $\GG^{(\phi,1)}$ are $1/6$, 1/2, and 1/3, respectively. Scaling factors: $x^{(\phi,1)} = 3$; $w^{(\phi,1)} = 2$.}
\label{fig:micro}
\end{figure}

Due to the generalized flow conservation law for unicast traffic and the choice of micro packet sizes, the total number of incoming micro packets to a node equals the total number of outgoing micro packets from a node. In other words, every outgoing micro packet can be associated with an incoming micro packet, and they can be viewed to have the same \emph{identity}. The links that carry micro packets with the same identity form a path. Since the routing is efficient, the path is acyclic. The route $E_p$ can be viewed as a \emph{convex combination} of service chain paths. More precisely, let $p_k$ be the number of micro packets with different identities that travel through a path $T_k^{(c,\phi)}$. Let $p_k/z$ be the weight of $T_k^{(c,\phi)}$. Let $\TT^{(c,\phi)}$ denote the set of paths from $s_c^{(\phi,0)}$ to $d_c^{(\phi,M_\phi)}$. The number of packets that travel through a link $(u^{(\phi,i)}, v^{(\phi,i)})$ is
$$\sum_{T_k^{(c,\phi)} \in \TT^{(c,\phi)}} \frac{p_k}{z} w^{(\phi,i)} 1\{(u^{(\phi,i)}, v^{(\phi,i)}) \in T_k^{(c,\phi)}\}.$$
The number of packets that travel through a link $(u^{(\phi,i-1)}, u^{(\phi,i)})$ is
$$\sum_{T_k^{(c,\phi)} \in \TT^{(c,\phi)}} \frac{p_k}{z} x^{(\phi,i)} 1\{(u^{(\phi,i-1)}, u^{(\phi,i)}) \in T_k^{(c,\phi)}\}.$$ 

\emph{Remark: }Two micro packets with different identities are distinct (\ie they carry different information). Micro packets that have the same identity can be viewed to carry the same (raw) information. More specifically, we view that a function $(\phi,i)$ takes a micro packet as an input and then outputs a micro packet with the same identity. 
Equivalently, in $\GG^{(\phi)}$, when a micro packet travels through a link in $\GG^{(\phi,i-1)}$, a link that connects $\GG^{(\phi,i-1)}$ and $\GG^{(\phi,i)}$, and a link in $\GG^{(\phi,i)}$, it has the same identity with possibly different sizes on the three links. In Fig.~\ref{fig:micro}, each color represents an identity.

\emph{Multicast traffic: }
The sizes of a micro packet are determined in the same manner as in the unicast case. First, consider a node where there is no packet duplication. The total number of incoming micro packets equals the total number of outgoing micro packets, and every outgoing micro packet can be associated with an incoming micro packet with the same identity. Then, consider a node where some packets $P$ are duplicated. All the micro packets that form $P$ are duplicated as well. The number of duplicated micro packets is integral. A duplicated micro packet inherits the same identity as the original micro packet. Every micro packet follows the same route as the packet that contains the micro packet. Since each link carries an integer number of micro packets and the number of duplicated packets is integral at each node, a micro packet is never split, and there exists an assignment of identity to micro packets that interprets the generalized flow conservation and packet duplication law. Due to the efficient routing assumption, the micro packets that have the same identity never visit the same node in $\GG^{(\phi)}$ more than once (\ie, the route forms a directed graph where the in-degree of any node is one). Thus, the route of the micro packets that have the same identity form a service chain Steiner tree (arborescence). The multicast flow can be viewed as a convex combination of service chain Steiner trees.

\subsection{Stability of the virtual queues} \label{sc:ap_virtual}
\emph{Proof of Theorem \ref{th:strong}}:
We consider a quadratic Lyapunov function $L(\tilde Q(t)) = \sum_{(u,v) \in \EE}\tilde{Q}^2_{uv}(t) + \sum_{u \in \VV}\tilde{Q}^2_u(t)$. The following inequality holds for all $\tilde Q$.
\begin{eqnarray*}
  \tilde{Q}^2(t+1) &\leq& (\tilde{Q}(t) + A^{(c,\phi),\pi}(t) - \mu)^2 \\
  & \leq & \tilde{Q}^2(t) +  (A^{(c,\phi),\pi}(t))^2 + \mu^2 \\
  &+& 2\tilde{Q}(t)A^{(c,\phi),\pi}(t) - 2\tilde{Q}(t) \mu.
\end{eqnarray*}
Let $A_{uv}^{\pi}(t) = \sum_{(c,\phi) \in (\CC,\Phi)}A_{uv}^{(c,\phi),\pi}(t)$; $A_{u}^{\pi}(t) = \sum_{(c,\phi) \in (\CC,\Phi)}A_{u}^{(c,\phi),\pi}(t)$. The Lyapupov drift $\Delta^\pi(t)$ is upper bounded by
\begin{eqnarray}
  \Delta^\pi(t) & \overset{\text{def}}{=} & \mathbb{E}(L(\tilde{Q}(t+1)) - L(\tilde{Q}(t))|\tilde{Q}(t)) \nonumber\\
  &\leq& B + 2\sum_{(u,v)\in \EE}\tilde{Q}_{uv}(t) \Big(\mathbb{E}(A_{uv}^{\pi}(t)|\tilde{Q}(t))-\mu_{uv}\Big) \nonumber \\
  &&+ 2\sum_{u \in \VV}\tilde{Q}_{u}(t) \Big(\mathbb{E}(A_{u}^{\pi}(t)|\tilde{Q}(t))-\mu_{u}\Big), \nonumber \\
\end{eqnarray}
where
\begin{eqnarray*}
  B &=& \sum_{(u,v)\in \EE} \Big( \mathbb{E} (A_{uv}^{\pi}(t))^2 + \mu_{uv}^2 \Big)\\
  && + \sum_{u \in \VV} \Big( \mathbb{E} (A_{u}^{\pi}(t))^2 + \mu_{u}^2 \Big)\\
  & \leq & \sum_{(u,v)\in \EE} \sum_{(c,\phi)\in (\CC,\Phi)} \sum_{i \in \{0,\dots,M_\phi\}}(w^{(\phi,i)})^2 \mathbb{E}(A^{(c,\phi)}(t))^2 \\
  && + \sum_{u\in \VV} \sum_{(c,\phi)\in (\CC,\Phi)} \sum_{i \in \{1,\dots,M_\phi\}}(x^{(\phi,i)})^2 \mathbb{E}(A^{(c,\phi)}(t))^2 \\
  && + \sum_{(u,v)\in \EE}\mu_{uv}^2 + \sum_{u\in \VV}\mu_{u}^2.
\end{eqnarray*}
For finite second moment of exogenous arrivals $\mathbb{E}(A^{(c,\phi)}(t))^2$ and finite scaling factors, $B$ is finite.

We next prove that the drift $\Delta^\pi(t)$ is negative for sufficiently large $\tilde Q(t)$, by comparing $\Delta^\pi(t)$ with the drift of a randomized policy.
For any $\boldsymbol{\bar \lambda}$ in the interior of the capacity region, there exists $\epsilon > 0$ such that $(1+ \epsilon) \boldsymbol{\bar \lambda}  \in \Lambda(\GG,\CC,\Phi)$. The rate $\boldsymbol{\lambda} = (1+ \epsilon) \boldsymbol{\bar \lambda} $ satisfies the constraints (\ref{eq:decomp}), (\ref{eq:comm}) and (\ref{eq:comp}). Let $\lambda_k^{(c,\phi)} = (1+\epsilon)\bar \lambda_k^{(c,\phi)}$, $\forall k,c,\phi$.
\begin{eqnarray*}
(1+\epsilon)\bar \lambda^{(c,\phi)} = \sum_{T_k^{(c,\phi)} \in \TT^{(c,\phi)}} (1+\epsilon)\bar \lambda_k^{(c,\phi)}, ~~ \forall (c,\phi) \in (\CC, \Phi),\\
  \sum_{(k,i,c,\phi) \in \SSS_{uv}} \hspace{-2mm} w^{(\phi,i)} (1+\epsilon)\bar \lambda_k^{(c,\phi)} \leq \mu_{uv}, ~~~~~~~~~~~ \forall (u,v) \in \EE, \\
  \sum_{(k,i,c,\phi) \in \SSS_u} \hspace{-2mm} x^{(\phi,i)} (1+\epsilon)\bar \lambda_k^{(c,\phi)} \leq \mu_{u}, ~~~~~~~~~~~~~~~~~ \forall u \in \VV.
\end{eqnarray*}

The randomized policy routes each incoming commodity-$(c,\phi)$ packet along $T_k^{(c,\phi)} \in \TT_k^{(c,\phi)}$ with probability $\bar \lambda_k^{(c,\phi)} / \bar \lambda^{(c,\phi)}$, $\forall k, c, \phi$. The expected arrival rates to $\tilde Q_{uv}$ and $\tilde Q_u$ at every time $t$ are
$$ \E A^{\text{rand}}_{uv}(t) =  \sum_{(k,i,c,\phi) \in \SSS_{uv}} \hspace{-2mm} w^{(\phi,i)} \bar \lambda_k^{(c,\phi)} \leq \mu_{uv} / (1 + \epsilon); $$
$$ \E A^{\text{rand}}_{u}(t) =  \sum_{(k,i,c,\phi) \in \SSS_u} \hspace{-2mm} x^{(\phi,i)} \bar \lambda_k^{(c,\phi)} \leq \mu_{u} / (1 + \epsilon). $$

Recall Eq. (\ref{eq:mintree}). Upon the arrival of $A^{(c,\phi)}(t)$ commodity-$(c,\phi)$ packets, policy $\pi^*$ chooses a route $T^{(c,\phi),\pi^*}$ that minimizes
$$\min_{\pi} \sum_{(u,v)\in \EE} \tilde{Q}_{uv}(t)A_{uv}^{(c,\phi),\pi}(t) + \sum_{u\in \VV} \tilde{Q}_{u}(t)A_{u}^{(c,\phi),\pi}(t).$$
The randomized policy randomly chooses $T_k^{(c,\phi)} \in \TT_k^{(c,\phi)}$ which has an equal or larger weight. Conditional on queue lengths $\tilde Q(t)$, taking expectation over the random variable $A^{(c,\phi)}(t)$ and the random actions in the randomized policy,
\begin{eqnarray*}
  && \sum_{(u,v)\in \EE} \tilde{Q}_{uv}(t) \E (A_{uv}^{(c,\phi),\pi^*}(t)|\tilde{Q}(t)) \\
  && + \sum_{u\in \VV} \tilde{Q}_{u}(t) \E (A_{u}^{(c,\phi),\pi^*}(t)|\tilde{Q}(t)) \\
  \leq&& \sum_{(u,v)\in \EE} \tilde{Q}_{uv}(t)\E (A_{uv}^{(c,\phi),\text{rand}}(t)|\tilde{Q}(t)) \\
  &&+ \sum_{u\in \VV} \tilde{Q}_{u}(t) \E (A_{u}^{(c,\phi),\text{rand}}(t) | \tilde{Q}(t)).
\end{eqnarray*}
Summing over all commodity packets, for each link $(u,v)$,
\begin{eqnarray*}
  \E (A_{uv}^{\pi^*}(t)|\tilde{Q}(t))&=&\sum_{(c,\phi) \in (\CC, \Phi)} \E (A_{uv}^{(c,\phi),\pi^*}(t)|\tilde{Q}(t)), \\
  \E (A_{uv}^{\text{rand}}(t)|\tilde{Q}(t)) &=& \sum_{(c,\phi) \in (\CC, \Phi)} \E  (A_{uv}^{(c,\phi),\text{rand}}(t)|\tilde{Q}(t)).
\end{eqnarray*}
Similar equalities hold for $\E (A_{u}^{\pi^*}(t)|\tilde{Q}(t))$ and $\E (A_{u}^{\text{rand}}(t)|\tilde{Q}(t))$. Therefore, we obtain
\begin{eqnarray*}
   \sum_{(u,v)\in \EE} \tilde{Q}_{uv}(t)\E (A_{uv}^{\pi^*}(t)|\tilde{Q}(t)) + \sum_{u\in \VV} \tilde{Q}_{u}(t) \E (A_{u}^{\pi^*}(t)|\tilde{Q}(t)) \\
  \leq \sum_{(u,v)\in \EE} \tilde{Q}_{uv}(t) \E (A_{uv}^{\text{rand}}(t)|\tilde{Q}(t)) + \sum_{u\in \VV} \tilde{Q}_{u}(t) \E (A_{u}^{\text{rand}}(t)|\tilde{Q}(t)).
\end{eqnarray*}

The action of the randomized policy does not depend on the queue length $\tilde{Q}(t)$. Therefore, $\E (A_{uv}^{\text{rand}}(t)|\tilde{Q}(t)) = \E A_{uv}^{\text{rand}}(t)$ and $\E (A_{u}^{\text{rand}}(t)|\tilde{Q}(t)) = \E A_{u}^{\text{rand}}(t)$.
Let $\epsilon' = \frac{\epsilon}{1+\epsilon} \min(\mu_{uv}, \mu_{u})$. The drift of policy $\pi^*$ can be upper bounded by
\begin{eqnarray*}
  \Delta^{\pi^*}(t) \leq B &+& 2\sum_{(u,v)\in \EE}\tilde{Q}_{uv}(t) \Big(\mathbb{E}(A_{uv}^{\pi^*}(t)|\tilde{Q}(t)) - \mu_{uv}\Big) \nonumber \\
  &+& 2\sum_{u \in \VV}\tilde{Q}_{u}(t) \Big(\mathbb{E}(A_{u}^{\pi^*}(t)|\tilde{Q}(t)) - \mu_{u}\Big) \nonumber \\
  \leq B &+& 2\sum_{(u,v) \in \EE} \tilde{Q}_{uv}(t)\Big(\E A_{uv}^{\text{rand}}(t) - \mu_{uv}\Big)\\
  ~~~& +& 2\sum_{u \in \VV} \tilde{Q}_{u}(t)\Big(\E A_{u}^{\text{rand}}(t) - \mu_{u}\Big)\\
  \leq B &-& 2\epsilon' \Big(\sum_{(u,v)\in \EE}\tilde{Q}_{uv}(t) + \sum_{u \in \VV}\tilde{Q}_{u}(t) \Big).
\end{eqnarray*}

Taking expectation over the virtual queue lengths $\tilde Q(t)$,
\begin{eqnarray}\label{eq:strong}
  \hspace{-5mm}\E L(\tilde Q(t+1)) &-& \E L(\tilde Q(t)) \leq B \nonumber\\
  &-& 2\epsilon' \Big(\sum_{(u,v)\in \EE} \E \tilde{Q}_{uv}(t) + \sum_{u \in \VV} \E \tilde{Q}_{u}(t)\Big).
\end{eqnarray}

Summing Eq.~(\ref{eq:strong}) from $t = 0,\dots,T-1$, and noting that $L(\tilde Q(T)) \geq 0$, $L(\tilde Q(0)) = 0$, we obtain
$$ \frac{1}{T} \sum_{t=0}^{T-1} \Big(\sum_{(u,v) \in \EE} \E \tilde{Q}_{uv}(t) + \sum_{u \in \VV} \E \tilde{Q}_{u}(t) \Big) \leq \frac{B}{2\epsilon'}. $$

By taking limsup on both sides, we have proved that all the virtual queues are strongly stable.

\subsection{Stability of the physical queues} \label{sc:ap_phy}
Before the proof, we first discuss the the intuitions on what makes a queue unstable and why the extended nearest-to-origin (ENTO) scheduling policy stabilizes the queue. Consider external packets arriving at a network, each of which has a specified path to travel. If the rate of external arrivals that will use link $e$ is no more than the service rate of $e$, the only cause of instability of the queue at $e$ is the variation of packet delays before reaching $e$. The packets may take different paths and experience different queueing delays. Within some time period, the actual arrival rate to $e$ can be higher than the service rate of $e$. The rate increase can be viewed as the contribution from the old packets in other queues (in contrast with the fresh packets that just arrived). The ENTO policy gives a higher priority to a packet that has traveled a smaller number of hops. Thus, few packets that have traveled a small number of hops are queued. These packets do not contribute much to the actual arrival rate to a subsequent queue. Thus, few packets that have traveled a slightly more number of hops are queued, because the only old packets that have higher priorities are those packets that have traveled a smaller number of hops. By induction, not many packets are in each queue, regardless of the number of hops that they have traveled, and thus the queues are stable.

In the following, we first show that, within any time interval, the packets that arrive at the network do not contribute to a physical link $e$ much more traffic than what can be transmitted through $e$. Then we prove that ENTO stabilizes the queue. The first part of proof is identical to \cite{umw_info17}. The second part of proof is similar, but takes care of flow scaling and cyclic routes.

\subsubsection{Average arrival rate is no more than the service rate for every physical queue}
For simplicity, we augment each computation node in the original graph $\GG$ by a self-loop that represents the computation queue. We denote the set of all links and self-loops by $\bar \EE$.

Since the virtual queues are strongly stable under policy $\pi^*$ (Theorem \ref{th:strong}), all the virtual queues are rate stable (Lemma 1 in \cite{umw_info17}).
\begin{eqnarray*}\label{eq:rateStable}
  \lim_{t \rightarrow \infty} \frac{\tilde{Q}_{e}(t)}{t} &=& 0, ~~\forall (u,v) \in \bar \EE, ~~\text{w.p.}~1.
\end{eqnarray*}

Almost surely for any sample path $\omega \in \Omega$ (\ie a realization of random arrivals),
\begin{equation}\label{eq:rate}
  A_e(\omega; t_0, t) \leq S_e(\omega; t_0, t) + F_e(\omega; t),~~ e \in \bar \EE,
\end{equation}
where $A_e(\omega; t_0, t) = \sum_{\tau = t_0}^{t - 1} A_e^{\pi^*}(\omega; \tau)$ is the total number of packets that arrive at virtual queue $\tilde Q_e$ during time $[t_0,t)$ under policy $\pi^*$ and sample path $\omega$; $S_e(\omega; t_0,t) = \sum_{\tau = t_0}^{t - 1} \mu_e = (t - t_0) \mu_e$ is the total number of packets that can be served by $e$; $F_e (\omega; t) = o(t)$ (\ie $\lim_{t\rightarrow \infty} F_e(\omega; t)/t = 0$).
Eq.~(\ref{eq:rate}) implies that the average arrival rate to the virtual queue $\tilde Q_e$ is no more than the service rate of $e$.

Next, we relate the arrival rate at a virtual queue to the arrival rate at a physical queue.
Since the routing policy for the physical system is identical to the routing policy $\pi^*$ for the virtual system, the exogenous packets that arrive at the network at time $t$ contribute a total of $A_e(\omega; t_0, t)$ packets to $e$ during the course of their service in the physical system. (Recall that a packet with a scaled size enters a virtual queue of a link immediately if the link is part of its route.) 

\subsubsection{ENTO stabilizes the physical queues}
We aim to prove that ENTO stabilizes the physical queues for any sample path $\omega$ that satisfies Eq.~(\ref{eq:rate}). In particular, we aim to prove
\begin{equation}\label{eq:rateStableReal}
  \lim_{t\rightarrow \infty} \frac{Q_e(\omega; t)}{t} = 0, ~~ \forall e \in \bar \EE.
\end{equation}
Then, ENTO stabilizes the physical queues almost surely because Eq.~(\ref{eq:rate}) holds for almost all sample paths.
$$ \lim_{t\rightarrow \infty} \frac{Q_e(t)}{t} = 0,~~\text{w.p.}~1,~\forall e \in \bar \EE.$$

For simplicity of presentation, we drop the $\omega$ in the notations and focus on one sample path. It has been shown in \cite{umw_info17} that there exists a \emph{non-decreasing non-negative} function $M(t) = o(t)$ such that
\begin{equation}\label{eq:rate2}
  A_e(t_0,t) \leq S_e(t_0,t) + M(t), ~~\forall e \in \bar \EE, t_0 \leq t.
\end{equation}

We introduce a few new notations. A \emph{hop-$k$} packet is a packet that has traveled $k$ hops from the origin. The processing at a computation node is also considered as one additional hop. A duplication of a packet inherits the hop of the original packet. The packets entering the network during $[t_0,t)$ contribute to $e$ a total of $A_e(t_0,t)$ packets. Among these packets, $A_e^k(t_0,t)$ packets use $e$ as their $(k+1)$-th hop, and they are hop-$k$ packets whiling waiting to cross $e$. Let $M^{\max} = \max_\phi M_\phi + 1$ denote the maximum number of functions in any service chain plus one. The maximum number of hops that a packet travels under the routing policy $\pi^*$ is $nM^{\max}$, where $n$ is the number of nodes in $\GG$. By definition,
$$A_e(t_0,t) = \sum_{k=0}^{nM^{\max}-1} A_e^k(t_0,t).$$
Let
$$\gamma = \max_{\phi \in \Phi, 0 \leq i \leq j \leq M_\phi} \frac{\max(w^{(\phi,j)}, x^{(\phi,j)})}{\min(w^{(\phi,i)}, x^{(\phi,i)})}$$
denote the maximum aggregated scaling factor. \Ie each packet that departs from any link contributes to at most $\gamma$ packets to any subsequent link in $\GG^{(\phi)}$. Note that the value $A_e(t_0,t)$ has taken the scalings into consideration, \ie Eqs.~(\ref{eq:scaling1}) and (\ref{eq:scaling2}). Let $Q_e(t)$ denote the physical queue length at $e$ at time $t$. Let $Q_e^k(t)$ denote the number of hop-$k$ packets in the queue at $e$ at time $t$. Let $Q^k(t) = \sum_{e \in \bar \EE}Q_e^{k}$ denote the total number of hop-$k$ packets in the network at time $t$. We prove by induction that $Q^k(t) = o(t)$ for all $k \in \{0,\dots,nM_{\max}-1\}$. 

\emph{Base step $k = 0$: }
Let $t_0 < t$ be the largest time at which no hop-0 packet were waiting to cross a specified link $e$. If no such time exists, $t_0 = 0$. During $[t_0,t)$, at most $A_e^0(t_0,t) \leq A_e(t_0,t) \leq S_e(t_0,t) + M(t)$ hop-0 packets arrived at $e$, by Eq.~(\ref{eq:rate2}). Moreover, $e$ is constantly transmitting hop-0 packets, for a total of $S_e(t_0,t)$ packets, because hop-0 packets have the highest priority and there are always hop-0 packets waiting to cross $e$ by the choice of $t_0$. Therefore,
$$ Q_e^0(t) \leq S_e(t_0,t) + M(t) - S_e(t_0,t) = M(t).$$
There are at most $\bar m = |\EE|+|\VV|$ physical queues. Therefore, $Q^0(t) \leq \bar m M(t)$. Let $B^0(t) = \bar m M(t) = o(t)$. Note that $B^0(t)$ is non-decreasing in $t$.

\emph{Induction step: }
Suppose that $Q^j(t) \leq B^j(t)$ for all $0 \leq j < k$, where $B^j(t) = o(t)$ is non-decreasing. We aim to prove that $Q^k(t) \leq B^k(t)$, for a non-decreasing $B^k(t) = o(t)$. Let $t_0$ be the largest time at which no hop-$k$ packets were waiting to cross a specified link $e$. Let $t_0 = 0$ if no such time exists.

The \emph{new} packets that arrive at the network during $[t_0,t)$ contributes at most $A_e^k(t_0,t)$ hop-$k$ packets to $e$ by time $t$. The \emph{old} packets that were already in the network by time $t_0$ contributes to $e$ at most $\gamma \sum_{0\leq j<k} B^j(t_0)$ hop-$k$ packets, because each of the $\sum_{0\leq j<k} B^j(t_0)$ old packets of hop fewer than $k$ contributes at most $\gamma$ hop-$k$ packets to $e$. Note that the old packets of hop more than $k$ never become hop-$k$ packets again. 

Next we bound the number of packets of hop fewer than $k$ that are transmitted through $e$ during $[t_0,t)$. The new packets that arrive at the network during $[t_0,t)$ contribute to $e$ at most $\sum_{0 \leq j<k} A_e^j(t_0,t)$ packets of hop fewer than $k$. Each old packet contributes at most $\gamma$ hop-$j$ packets ($0 \leq j< k$). Thus, the total number of packets of hop fewer than $k$ contributed by one old packet is at most $\gamma k$. For a total of $\sum_{0\leq j<k} B^j(t_0)$ old packets, at most $\gamma k \sum_{0\leq j<k} B^j(t_0)$ packets of hop fewer than $k$ travel through $e$ during $[t_0,t)$.

The link is consistently processing packets of hop no more than $k$ during $[t_0,t)$, by the choice of $t_0$. The packets that have hop fewer than $k$ have a higher priority than the hop-$k$ packets.
Thus, the number of hop-$k$ packets that are processed by $e$ is at least
$ \max(0, S_e(t_0,t) - \sum_{0 \leq j<k} A_e^j(t_0,t) - \gamma k \sum_{0\leq j<k} B^j(t_0)).$

The number of hop-$k$ packets at queue $e$ at time $t$ is at most
\begin{eqnarray*}
  Q_e^k(t) & \leq & A_e^k(t_0,t) + \gamma \sum_{0\leq j<k} B^j(t_0) \\
   &-& (S_e(t_0,t) - \sum_{0 \leq j<k} A_e^j(t_0,t) - \gamma k \sum_{0\leq j<k} B^j(t_0)) \\
   &\leq& \gamma(k+1) \sum_{0\leq j<k} B^j(t_0) + M(t).
\end{eqnarray*}

Let $B_e^k(t) = \gamma (k+1) \sum_{0\leq j<k} B^j(t) + M(t)$. Since $M(t)$ and $B^j(t)$ are non-decreasing in $t$ for $0 \leq j < k$, $B_e^k(t)$ is a non-decreasing function and $B_e^k(t) \geq \gamma(k+1) \sum_{0\leq j<k} B^j(t_0) + M(t) \geq Q_e^k(t)$. Since $B^j(t) = o(t)$ for $0 \leq j < k$ and $M(t) = o(t)$, we have $B_e^k(t) = o(t)$. Let $B^k(t) = \sum_{e \in \bar \EE} B_e^k(t) = \bar m B_e^k(t)$. It is easy to check that $B^k(t) = o(t)$ is a non-decreasing function.

We have proved that $Q^k(t) = o(t)$ for all $k$. Then, the sum of all queue lengths $\sum_{e \in \bar \EE} Q_e(t) = \sum_{k} Q^k(t) \leq \sum_{k}B^k(t) = o(t)$. Therefore, all the physical queues are stable, and Eq.~(\ref{eq:rateStableReal}) holds.


\begin{thebibliography}{1}


\bibitem{fxbook} M. Weldon,
``The future X network: a Bell Labs perspective,''
\emph{CRC Press}, October 2015.

\bibitem{industrial_internet}
Industrial Internet Consortium,
https://www.iiconsortium.org/

\bibitem{ar_2013}
 A. B. Craig,
 ``Understanding augmented reality: concepts and applications,"
 \emph{Newnes}, 2013.

\bibitem{charikar2014multi}  M. Charikar, Y. Naamad, J. Rexford, and K. Zou,
``Multi-commodity flow with in-network processing,'' arXiv preprint arXiv:1802.09118, 2018.

\bibitem{csdp_icc15} M. Barcelo, J. Llorca, A. M. Tulino, N. Raman,
``The cloud servide distribution problem in distributed cloud networks,''
\emph{Proc. IEEE ICC}, 2015.

\bibitem{bari2015orchestrating} M. Bari, S. R. Chowdhury, R. Ahmed, R. Boutaba,
``On orchestrating virtual network functions in NFV,''
\emph{Proc. 11th International Conference on Network and Service Management (CNSM)}, 2015.

\bibitem{barcelo2016} M. Barcelo, A. Correa, J. Llorca, A. Tulino, J. Lopez, A. Morell, ``IoT-Cloud Service Optimization in Next Generation Smart Environments,'' \emph{IEEE J. Selected Areas in Comm.}, vol. 34, no. 12, pp. 4077--4099, October 2016.

\bibitem{Cohen15} R. Cohen, L. Lewin-Eytan, J.S. Naor, D. Raz,
``Near optimal placement of virtual network functions,''
\emph{Proc. IEEE INFOCOM}, 2015.

\bibitem{Cao2017} Z. Cao, S. S. Panwar, M. Kodialam, and T. V. Lakshman,
``Enhancing mobile networks with software defined networking and cloud computing,"
{\em IEEE/ACM Trans. Networking}, vol. 25, no. 3, pp. 1431-1444, June 2017.

\bibitem{nsdp_info17} H. Feng, J. Llorca, A. M. Tulino, D. Raz, A. F. Molisch,
``Approximation algorithms for the NFV service distribution problem,''
\emph{Proc. IEEE INFOCOM}, 2017.

\bibitem{Kuo2017}
J. Kuo, S. Shen, H. Kang, D. Yang, M. Tsai, and W. Chen,
``Service chain embedding with maximum flow in software defined network and application to the next-generation cellular network architecture",
\emph{Proc. IEEE INFOCOM}, 2017.
\bibitem{dcnc_info16} H. Feng, J. Llorca, A. M. Tulino, A. F. Molisch, ``Dynamic network service optimization in distributed cloud networks,''
{\em Proc. IEEE INFOCOM SWFAN Workshop}, April 2016.

\bibitem{dcnc_icc16} H. Feng, J. Llorca, A. M. Tulino, A. F. Molisch,
``Optimal dynamic cloud network control,''
\emph{IEEE/ACM Trans. Networking}, vol. 26, no. 5, pp. 2118--2131, Sept. 2018.

\bibitem{Destounis2016}
A. Destounis, G. Paschos, I. Koutsopoulos,
``Streaming big data meets backpressure in distributed network computation,"
{\em Proc. IEEE INFOCOM}, April 2016.

\bibitem{umw_info17} A. Sinha, E. Modiano,
``Optimal control for generalized network-flow problems,''
{\em IEEE/ACM Trans. Networking}, vol. 26, no. 1, pp. 506-519, Feb 2018.

\bibitem{conf}
J. Zhang, A. Sinha, J. Llorca, A. Tulino, E. Modiano,
``Optimal control of distributed computing networks with mixed-cast traffic flows,"
{\em Proc. IEEE INFOCOM}, April 2018.

\bibitem{zhao2010unified} H.C. Zhao, C.H. Xia, Z. Liu, and D. Towsley,
``A unified modeling framework for distributed resource allocation of general fork and join processing networks,''
\emph{ACM SIGMETRICS Performance Evaluation
Review}, vol. 38, no. 1, pp. 299--310, 2010.

%
%



\bibitem{nto} D. Gamarnik, ``Stability of adaptive and non-adaptive packet routing policies in adversarial queueing networks,''
\emph{SIAM J. Comput.}, Vol. 32, No. 2, pp. 371--385, 2003.

\end{thebibliography}
\end{document}